\documentclass[11pt,a4paper]{article}
\usepackage{vmargin}
\setmarginsrb{1.0in}{1.0in}{1.0in}{1.0in}{0mm}{0mm}{5mm}{5mm}
\usepackage{graphicx}
\usepackage[utf8]{inputenc}
\usepackage{todonotes}
\usepackage{subcaption}
\usepackage{multirow}
\usepackage{booktabs}
\usepackage{ulem}
\usepackage[vlined,slide,linesnumbered,ruled]{algorithm2e}
\usepackage{float}
\usepackage{amsmath}
\usepackage{amssymb}
\usepackage{hyperref}
\usepackage{amsthm}

\newtheorem{theorem}{Theorem}

\newtheorem{conjecture}{Conjecture}

\newcommand{\pg}[0]{p_G}
\newcommand{\npg}[0]{N_{\pg}}
\newcommand{\ra}[1]{\renewcommand{\arraystretch}{#1}}

\title{Convergence and Hardness of Strategic Schelling Segregation\\{\small (full version)}}
\author{Hagen Echzell\thanks{Hasso Plattner Institute, University of Potsdam, Germany, \texttt{firstname.lastname@student.hpi.de}} \and Tobias Friedrich\thanks{Hasso Plattner Institute, University of Potsdam, Germany, \texttt{firstname.lastname@hpi.de}} \and Pascal Lenzner\footnotemark[2] \and Louise Molitor\footnotemark[2] \and Marcus Pappik\footnotemark[1] \and Friedrich Schöne\footnotemark[1] \and Fabian Sommer\footnotemark[1] \and David Stangl\footnotemark[1]}

\date{}

\begin{document}

\maketitle             
\begin{abstract}
\noindent The phenomenon of residential segregation was captured by Schelling's famous segregation model where two types of agents are placed on a grid and an agent is content with her location if the fraction of her neighbors which have the same type as her is at least $\tau$, for some $0<\tau<1$. Discontent agents simply swap their location with a randomly chosen other discontent agent or jump to a random empty cell.

We analyze a generalized game-theoretic model of Schelling segregation which allows more than two agent types and more general underlying graphs modeling the residential area. For this we show that both aspects heavily influence the dynamic properties and the tractability of finding an optimal placement.
We map the boundary of when improving response dynamics (IRD), i.e., the natural approach for finding equilibrium states, are guaranteed to converge. For this we prove several sharp threshold results where guaranteed IRD convergence suddenly turns into the strongest possible non-convergence result: a violation of weak acyclicity. In particular, we show such threshold results also for Schelling's original model, which is in contrast to the standard assumption in many empirical papers. 
Furthermore, we show that in case of convergence, IRD find an equilibrium in $\mathcal{O}(m)$ steps, where $m$ is the number of edges in the underlying graph and show that this bound is met in empirical simulations starting from random initial agent placements.

\end{abstract}
\section{Introduction}
Residential segregation is a well-known and remarkable phenomenon in many major metropolitan areas. There, local and myopic location choices by many individuals with preferences over their direct residential neighborhood yield cityscapes which are severely segregated along racial and ethnical lines (see Fig.~\ref{fig:realworld}(a)). 
\begin{figure}[h]
	\centering
	\begin{subfigure}{0.34\textwidth}
		\centering
		\includegraphics[height=4.2cm]{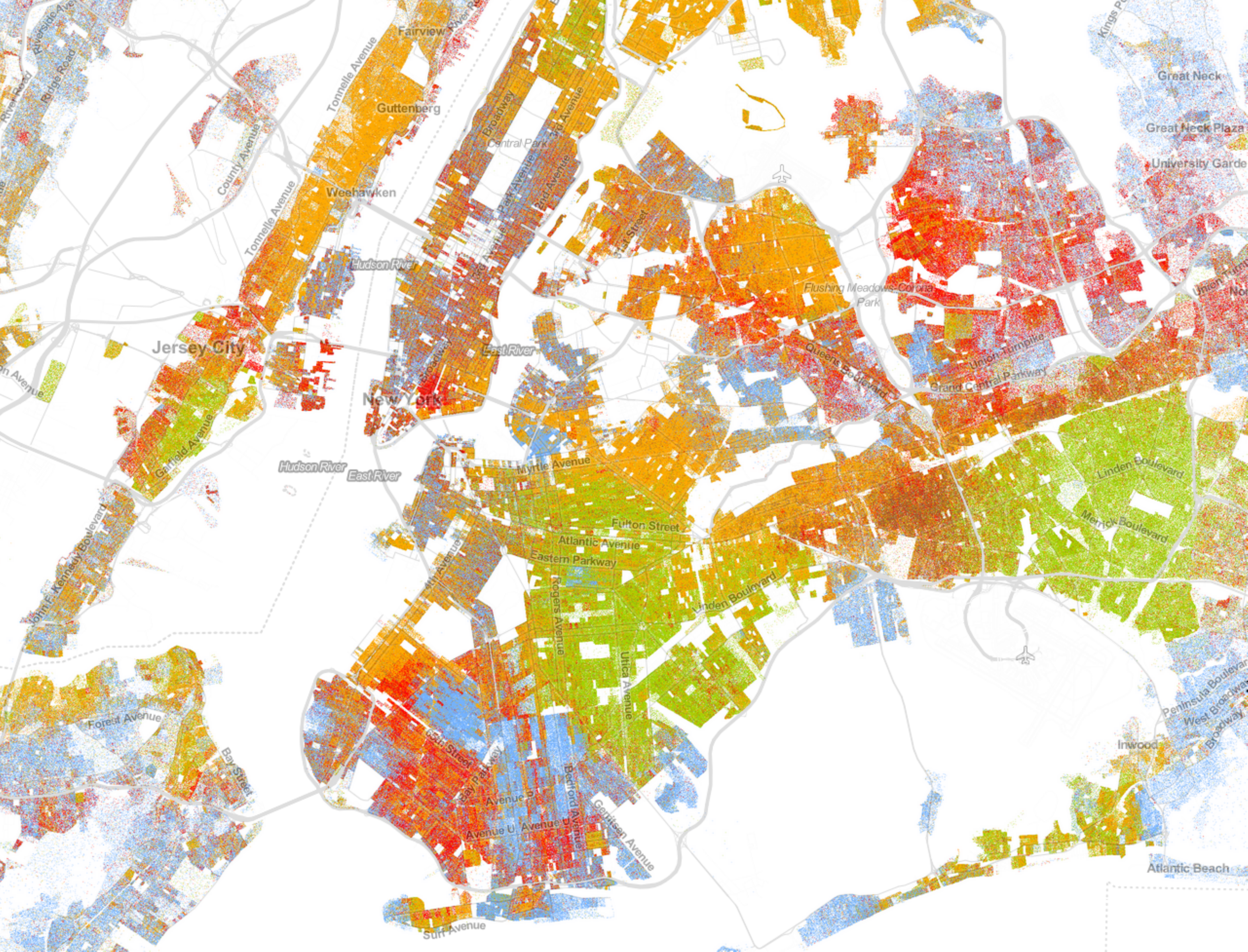}
		\caption{New York City}
	\end{subfigure}
	\begin{subfigure}{0.32\textwidth}
		\centering
		\includegraphics[height=4.2cm]{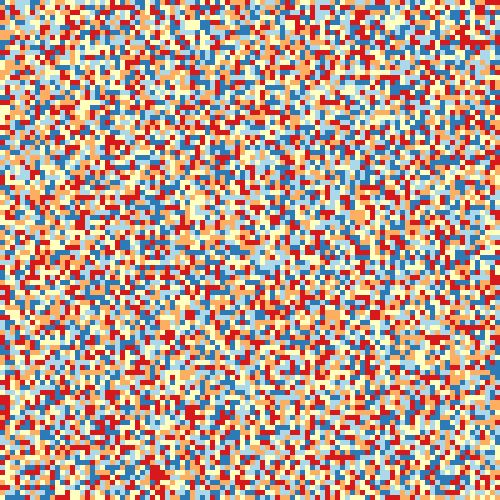}
		\caption{Random grid placement}
	\end{subfigure}
	\begin{subfigure}{0.32\textwidth}
		\centering
		\includegraphics[height=4.2cm]{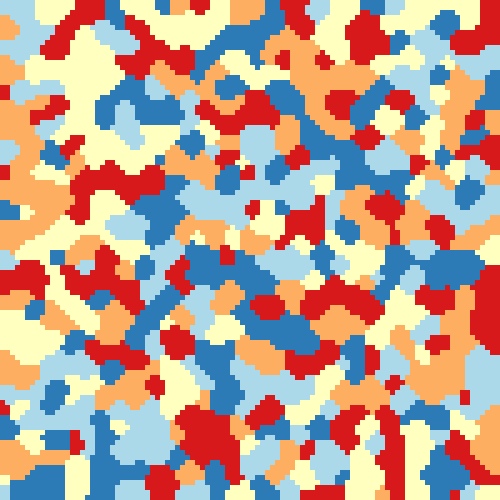}
		\caption{Stable grid placement}
	\end{subfigure}
	\caption{(a) Residential segregation in New York City, color-coded by ethnicity. Every dot corresponds to a citizen. Shown is a snippet from the Racial Dot Map~\cite{C13} based on data from the 2010 US Census. (b) Initial random placement on a grid in Schelling's model. (c) Equilibrium found for the instance in (b) with $\tau=\frac{1}{2}$ via improving response dynamics. \label{fig:realworld} }
\end{figure}
Hence, local strategic choices on the micro level lead to an emergent phenomenon on the macro level. This paradigm of ``micromotives'' versus ``macrobehavior''~\cite{S06} was first investigated and modeled by Thomas Schelling who proposed a very simple stylized model for analyzing residential segregation~\cite{Sch69,Schelling71}. With the use of two types of coins as two types of individual agents and graph paper serving as residential area, Schelling demonstrated the emergence of segregated neighborhoods under the simple assumption of the following threshold behavior: agents are content with their current location if the fraction of agents of their own type in their neighborhood is at least~$\tau$, where $0 < \tau < 1$ is a global parameter which applies to all agents. 
Content agents do not move, but discontent agents will swap their location with some other random discontent agent or perform a random jump to an unoccupied place. Given this, Schelling demonstrated by experiment that starting from a uniformly random distribution of the agents (see Fig.~\ref{fig:realworld}(b)) the induced random process yields a residential pattern which shows strong segregation (see Fig.~\ref{fig:realworld}(c)). While this is to be expected for intolerant agents, i.e., $\tau > \frac{1}{2}$, the astonishing finding of Schelling was that this also happens for tolerant agents, i.e., $\tau \leq \frac{1}{2}$. This counter-intuitive observation explains why even in a very tolerant population segregation along racial/ethnical, religious or socio-economical lines can emerge. 

Schelling's elegant model became one of the landmark models in sociology and it spurred a significant number of research articles which studied and motivated variants of the model, e.g. the works by Clark~\cite{Cla86}, Alba \& Logan~\cite{RAL93}, Benard \& Willer~\cite{BW07}, Henry et al.~\cite{HPP11} and Bruch~\cite{Bru14}, to name only a few. Interestingly, also a physical analogue of Schelling's model was found by Vinković \& Kirman~\cite{Vin06} but it was argued by Clark \& Fosset~\cite{clark2008understanding} that such models do not enhance the understanding of the underlying social dynamics. In contrast, they promote simulation studies via agent-based models where the agents' utility function is inspired by real-world behavior. Schelling's model as an agent-based system can be easily simulated on a computer and many such empirical simulation studies were conducted to investigate the influence of various parameters on the obtained segregation, e.g. see the works by Fossett~\cite{fossett2006ethnic}, which use the simulation framework SimSeg~\cite{fossett1998simseg}, Epstein \& Axtell~\cite{EA96}, Gaylord \& d'Andria~\cite{gaylord1998simulating}, Pancs \& Vriend~\cite{Pan07}, Singh et al.~\cite{SVW09} and Benenson et al.~\cite{IEE09}.  

All these empirical studies consider essentially an induced random process, i.e., that discontent agents are activated at random and active agents then swap or jump to other randomly selected positions. 
In some frameworks, like SimSeg~\cite{fossett1998simseg} or the model by Pancs \& Vriend~\cite{Pan07}, agents only change their location if this yields an improvement according to some utility function. This assumption of having rational agents which act strategically matches the behavior of real-world agents which would only move if this improves their situation.

This paper sets out to explore the properties of such strategic dynamic processes and the tractability of the induced optimization problems. 

\subsection{Related Work}
Recently, a series of papers by Young~\cite{You98}, Zhang~\cite{Zha04,Zha04b}, Gerhold et al.~\cite{Gerhold08}, Brandt et al.~\cite{BIK12,BIK17}, Barmpalias et al.~\cite{BEL14,BEL16} and Bhakta et al.~\cite{Bhakta14} initiated a rigorous analysis of stochastic processes induced by Schelling's model. In these processes either two randomly chosen unhappy agents of different type swap positions~\cite{You98,Zha04,Zha04b} or a randomly chosen agent changes her type with a certain probability~\cite{BIK12,BEL14,Bhakta14,BEL16,BIK17}. It is worth noticing that both types of processes are closely related but not identical to Schelling's original model where discontent agents move to different positions until they become content with their current location.
The focus of the above mentioned works is on investigating the expected size of the obtained homogeneous regions, but it is also shown that the stochastic processes starting from a uniform random agent placement converge with high probability to a stable placement. The convergence time was considered by Mobius \& Rosenblat~\cite{Mobius00} who observe that the Markov chain analyzed in~\cite{You98,Zha04,Zha04b} has a very high mixing time. Bhakta et al.~\cite{Bhakta14} show in the two-dimensional grid case a dichotomy in mixing times for high $\tau$ and very low $\tau$~values.

To the best of our knowledge, only a few papers have investigated game-theoretic models of Schelling segregation. Pancs \& Vriend~\cite{Pan07} used different types of utility functions for their agents in extensive simulation experiments. On the theory side, Zhang~\cite{Zha04b,zhang2011tipping} analyzed a model where the agents are endowed with a noisy single peaked utility function, which is a departure from the threshold behavior proposed by Schelling. Grauwin et al.~\cite{Grauwin12} generalized the results. In contrast, 
the behavior of the original model is closely captured by a game-theoretic model which was proposed in a recent paper by Chauhan et al.~\cite{CLM18}. They employ a utility function which depends on the type ratio in the neighborhood and which increases linearly with the fraction of agents of the own type in the neighborhood until a fraction of $\tau$ is reached. 
The authors of~\cite{CLM18} investigate the convergence behavior of the induced sequential game for the cases where discontent agents are restricted either to performing only improving location swaps (called the Swap Schelling Game (SSG)) or where discontent agents are only allowed to jump to empty locations to improve on their situation (called the Jump Schelling Game (JSG)). This corresponds to analyzing IRD, whose analysis is also our main contribution.
Their main result is a proof that IRD in the SSG converge if $\tau \leq \frac{1}{2}$ for any underlying connected graph as residential area. If the underlying graph is regular then IRD convergence is guaranteed for arbitrary $\tau$. For the JSG they prove guaranteed IRD convergence on 2-regular graphs. We improve on these results in various ways by exactly characterizing when IRD convergence is ensured. 
In~\cite{CLM18} also an extension of Schelling's model is considered, where agents also have preferences over the different locations in the residential area and agents additionally strive for being close to their favorite position. For this augmented version, they show for the JSG with $\frac{1}{3} < \tau \leq \frac{2}{3}$ that improving response cycles exist, i.e., that convergence is not guaranteed.  

Very recently, Elkind et al.~\cite{elkind19} studied a variant of the model by Chauhan et al.~\cite{CLM18}, where the agents are partitioned into stubborn and strategic agents. The former agents do not move and the latter agents try to maximize the fraction of same-type agents in their neighborhood by jumping to a suitable empty location. This corresponds to a variant of the JSG with $\tau = 1$. They show that equilibria are not guaranteed to exist and that deciding equilibrium existence or the existence of an agent placement with certain social welfare is NP-hard. 
This relates to our hardness results for computing socially optimal states. 
They also prove that the price of anarchy and the price of stability can be unbounded.  

All mentioned works, with SimSeg~\cite{fossett1998simseg} and the work by Elkind et al.~\cite{elkind19} as exceptions, assume that exactly two types of agents exist. 
In SimSeg and~\cite{elkind19}, agents only differentiate between agents of their own type and agents of other types. As we will discuss later, this is a very restricted point of view and this will correspond to our ``one-versus-all'' version.   

\subsection{Model and Notation}
We consider a network $G = (V, E)$, where $V$ is the set of nodes and $E$ is the set of edges,
which is connected, unweighted and undirected. The network $G$ serves as the underlying graph modeling the residential area in which the agents will select a location. If every node in $G$ has the same degree $\Delta$,
i.e., the same number of incident edges, then we say that $G$ is a \textit{$\Delta$-regular graph}. Let $\textit{deg}_G(v)$ be the degree of a node $v \in V$ in $G$ and for a given node $u \in V$ let $\Gamma_G(u)$ denote the set of nodes $v \neq u$ so that an edge $\{u,v\}$ exists in $E$. We call $\Gamma_G(u)$ the \textit{neighborhood} of $u$ in network~$G$.
Let $A$ be the set of agents and $P(A) = \{T_1, T_2,\dots,T_k\}$ be any partition of $A$ into $k$ non-empty distinct sets, called \textit{types}, which model racial/ethnic, religious or socio-economic groups. For $k = 2$ this corresponds to Schelling's original model~\cite{Sch69,Schelling71} with two different types of agents.
Let $t: A \mapsto P(A)$ be a surjective function such that $t(a) = T$ if $a \in T$. We say that agent $a$ is of type~$t(a)$.
A state of our games is defined by an injective \textit{placement} $\pg: A \mapsto V$ which assigns every agent to a node in the network $G$ and we call $\pg(a)$ \textit{agent $a$'s location under placement $\pg$}. Two agents $a,b \in A$ are \textit{neighbors under placement $\pg$} if $\pg(b) \in \Gamma_G(\pg(a))$ and we denote the set of neighbors of $a$ under placement $\pg$ as $\npg(a)$.
For any agent $a \in A$, we define $N_{\pg}^T(a) = \{b \in T \mid b \in \npg(a) \}$, as the set of agents of type $T$ in the neighborhood of agent~$a$ under placement $\pg$. 

For any agent $a \in A$ in a placement $\pg$, we define \textit{agent $a$'s positive neighborhood} $\npg^+(a)$ as $\npg^{t(a)}(a)$. 
For \textit{agent $a$'s negative neighborhood}, we define two different versions, called the \textit{one-versus-all} and \textit{one-versus-one} versions. In the one-versus-all version an agent wants a certain fraction of agents of her own type in her neighborhood, regardless of the specific types of neighboring agents with other types, so $\npg^-(a)$ is $\npg(a)\setminus \npg^+(a)$. In contrast to this, in the one-versus-one version an agent only compares the number of own-type agents to the number of agents in the largest group of agents with different type in her neighborhood.
Thus, we define the negative neighborhood of an agent $a$ under placement $\pg$ as the set of neighboring agents of the type $T \neq t(a)$ that make up the largest proportion among all neighbors, i.e., $\npg^- (a) = \npg^T(a)$ such that $T \in P(A) \setminus \{t(a)\}$ and $|\npg^T(a)| \geq |\npg^{T'}(a)|$ for all $T' \in P(A) \setminus \{t(a)\}$.
Notice that the one-versus-all and one-versus-one version coincide for $k=2$, thus both versions generalize the two type case.
If an agent $a$ has no neighboring agents, i.e., $\npg(a) = \emptyset$, we say that $a$ is \textit{isolated}, otherwise $a$ is \textit{un-isolated}.

Let $\tau \in (0,1)$ be the \textit{intolerance parameter}. Similar to Schelling's model we say that an agent $a$ is \textit{content} with placement $\pg$ if agent $a$ is un-isolated and at least a $\tau$-fraction of the agents in agent $a$'s positive and negative neighborhood under $\pg$ are in agent $a$'s positive neighborhood. Hence, agent~$a$ is content if she is un-isolated and $\frac{|\npg^+(a)|}{|\npg^+(a)| + |\npg^-(a)|} \geq \tau$, otherwise $a$ is \textit{discontent} with placement~$\pg$.
We call the ratio $\textrm{pnr}_{\pg}(a) = \frac{|\npg^+(a)|}{|\npg^+(a)| + |\npg^-(a)|}$ the \textit{positive neighborhood ratio} of agent~$a$.
An agent's aim is to find a node in the given network where she is content or, if this is not possible, where she has the highest possible positive neighborhood ratio.
Therefore, and analogous to~\cite{CLM18}, we define the \textit{cost function} of an agent $a$ in a placement $\pg$ for network $G$ as follows: 
$$\textrm{cost}_{\pg}(a) = \begin{cases}
\max\{0, \tau - \textrm{pnr}_{\pg}(a)\}, &\textrm{ if }a\textrm{ is un-isolated}, \\
\tau, & \textrm{ if }a\textrm{ is isolated}.
\end{cases} 
$$
Thus, agent $a$ is content with placement $\pg$, if and only if $\textrm{cost}_{\pg}(a) = 0$.
The \textit{placement cost}, denoted $\textrm{cost}_{\pg}(A)$, of a placement $\pg$ in a network $G$ is simply the number of all discontent agents:
$\textrm{cost}_{\pg}(A) = | \{a \in A \mid \textrm{cost}_{\pg}(a) \neq 0\}|$.

\paragraph*{The Strategic Games:}
The \textit{strategy space} of an agent is the set of all nodes in the network $G$. 
An agent can change her strategy either via swapping with another agent who agrees or via jumping to another unoccupied node in network. This yields the \textit{Swap Schelling Game (SSG)} and the \textit{Jump Schelling Game (JSG)}.

For the SSG we will assume that all nodes of $G$ are occupied.
A \textit{location swap}, or \textit{swap}, of two agents $a, b \in A$ under placement $\pg$ is to exchange the occupied nodes of both agents. This yields a new placement $\pg'$ with $\pg'(a) = \pg(b)$, $\pg'(b) = \pg(a)$ and $\pg(x) = \pg'(x)$, for any other agent $x\in A\setminus\{a,b\}$. Two agents $a,b \in A$ would only agree to such a swap if it strictly decreases the cost of both agents, i.e., $\textrm{cost}_{\pg'}(a) < \textrm{cost}_{\pg}(a)$ and $\textrm{cost}_{\pg'}(b) < \textrm{cost}_{\pg}(b)$.
Hence, swapping agents are always of different types.
If for some placement $\pg$ no improving swap exists, then we say that $\pg$ is \textit{swap-stable}. 

In the JSG we assume that there exist empty nodes in the underlying graph and an agent can change her strategy to any currently empty node, which we denote as a \textit{jump} to that node. An agent will only jump to another empty node, if this strictly decreases her cost. An equilibrium placement in the JSG where no agent can improve via jumping is called \textit{jump-stable}.

If the game is clear from the context, we will simply say that a placement $\pg$ is \textit{stable}. If we have more than two different agent types we denote the one-versus-all version of the SSG and the JSG as \textit{$1$-$k$-SSG} and \textit{$1$-$k$-JSG}, respectively and the one-versus-one version of both games as \textit{$1$-$1$-SSG} and \textit{$1$-$1$-JSG}, respectively.

\paragraph*{Improving Response Dynamics and Potential Games:} 
We analyze whether \textit{improving response dynamics (IRD)}, i.e., the natural approach for finding equilibrium states where agents sequentially try to change towards better strategies until no agent can further improve, will converge.
For showing this we employ \textit{ordinal potential functions}. Such a function $\Phi$ maps placements to real numbers such that  if an agent (or a pair of agents) under placement $\pg$ can improve by a jump (or a swap) which results in placement $\pg'$ then $\Phi(\pg) > \Phi(\pg')$ holds. That is, any improving strategy change also decreases the potential function value. The existence of an ordinal potential function shows that a game is a \textit{potential game}~\cite{MS96}, which guarantees the existence of pure equilibria and that IRD must terminate in an equilibrium. In contrast, an \textit{improving response cycle (IRC)} is a sequence of improving strategy changes which visits the same state of the game twice. The existence of an IRC directly implies that a potential function cannot exist and thus, that IRD may not terminate. However, even with existing IRCs it is still possible, that from any state of the game there exists a finite sequence of improving strategy-changes which leads to an equilibrium. In this case the game is \textit{weakly acyclic}~\cite{Young93}. Thus, the strongest possible non-convergence result is a proof that a game is not weakly acyclic.

\subsection{Our Contribution}

Our main contribution is a thorough investigation of the convergence behavior of improving response dynamics in variants of Schelling's model. Previous work, including Schelling's original papers and all the mentioned empirical simulation studies, assume that IRD always converge to an equilibrium. We challenge this basic assumption by precisely mapping the boundary of when IRD are assured to find an equilibrium. 
We show that IRD behave radically different in the swap version compared to the jump version. Moreover, we show that this contrasting behavior can even be found within these two variants.  
We demonstrate the extreme cases of guaranteed IRD convergence, i.e., the existence of an ordinal potential function, and the strongest possible non-convergence result, i.e., that even weakly acyclicity is violated. For this, we provide sharp threshold results where for some $\tau^*$ IRD are guaranteed to convergence for $\tau \leq \tau^*$ and we have non-weak-acyclicity for $\tau > \tau^*$, depending on the underlying graph. See Table~\ref{tbl:previous_results}. 
\begin{table*}[t]
	\centering	
	\ra{1.3}
	\begin{tabular}{@{}lll r cc rrrr@{}}
		\toprule
		&  \multicolumn{1}{c}{ $1$-$k$-SSG}  &   \multicolumn{2}{c}{ $1$-$1$-SSG}  &  \multicolumn{2}{c}{ $1$-$k$-JSG}  &   \multicolumn{2}{c}{ $1$-$1$-JSG} \\
		\midrule
		reg. &\hspace*{+0.5cm}  \checkmark(Thm.\ref{thm:1nSSG_reg})&   \checkmark(Thm.\ref{thm:11SSG-convergence})  &  $\tau \leq \frac{1}{\Delta}$ & \hspace*{+0.1cm} \checkmark(Thm.\ref{thm:1nJSG}) &  $\tau \leq \frac{2}{\Delta}$  &  \hspace*{+0.1cm}\checkmark(Thm.\ref{thm:11JSG-convergence}) & $\tau \leq \frac{1}{\Delta}$  \\ 
		&& o (Thm.\ref{11SSG-IRC-reg}) &   $\tau \geq \frac{6}{\Delta}$  &\hspace*{+0.1cm} o\ (Thm.\ref{thm:1nJSG_BRC}) &   $\tau > \frac{2}{\Delta}$ & o\ (Thm.\ref{thm:11JSG_reg_IRC}) &$\tau > \frac{2}{\Delta}$\\[1 ex]
		arb. &  \checkmark\cite{CLM18} $k = 2$, $\tau \leq \frac12$&  $\times$(Thm.\ref{thm:11SSG-BRC})&&\hspace*{+0.2cm}$\times$(Thm.\ref{thm:1nJSG-IRC}) & & $\times$(Thm.\ref{thm:11JSG_IRC})\\
		&  $\times$(Thm.\ref{thm:2SSG_BRC}\&\ref{thm:1nSSG_BRC}) ow.& &&   & & \\
		
		\bottomrule
	\end{tabular}
	\vspace*{+0.1cm}
	\caption{Results regarding IRD. ``reg.'' stands for $\Delta$-regular graphs, ``arb'' for arbitrary graphs, which model the residential area. ``\checkmark'' denotes that IRD converge to an equilibrium, ``o'' denotes the existence of an IRC. ``$\times$'' denotes that the version is not weakly acyclic. If $\tau$ is omitted, the result holds for any $0<\tau<1$.}\label{tbl:previous_results}
\end{table*}

In case of IRD convergence, we show that this happens after $\mathcal{O}(|E|)$ many jumps/swaps on an underlying graph $G=(V,E)$. We show via experiments that instances with randomly chosen initial placements meet this upper bound. 

Besides analyzing IRD, we start a discussion about segregation models with more than two agent types. Besides the simple generalization of differentiating only between own type and other types, i.e., the $1$-$k$-SSG and $1$-$k$-JSG, we propose a more natural alternative, called the $1$-$1$-SSG and the $1$-$1$-JSG, where agents compare the type ratios only with the largest subgroup in their neighborhood. The idea here is that a minority group mainly cares about if there is a dominant other group within the neighborhood.

Moreover, we investigate the influence of the underlying graph
on the hardness of computing an optimal placement. We show that computing this is NP-hard for arbitrary underlying graphs if $\tau = \tfrac{1}{2}$ or if $\tau$ is close to the maximum degree in the graph. In contrast to this, we provide an efficient algorithm for computing the optimum placement on a $2$-regular graph with two agent types. The number of agent types also has an influence: we establish NP-hardness even on $2$-regular graphs if there are sufficiently many agent types. 

\section{Schelling Dynamics for the Swap Schelling Game \label{sec:twotypes}}
In the following section we analyze the convergence behavior of IRD for the strategic segregation process via swaps. Chauhan et al.~\cite{CLM18} already proved initial results in this direction, in particular that the SSG for two types of agents converges for the whole range of $\tau$, i.e $\tau \in (0,1)$, on $\Delta$-regular graphs and for $\tau \leq \frac12$ on arbitrary graphs. We close the gap and present a matching non-convergence bound in the SSG on arbitrary graphs.

The $1$-$k$-variant seems to be a straightforward generalization of the two type case. An agent simply compares the number of neighbors of her type with the total number of neighbors.
Interestingly, our IRD convergence results for the $1$-$k$-SSG with $k>2$ for arbitrary networks for $\tau \leq \frac12$ are in sharp contrast to the results for $k=2$: On arbitrary networks with tolerant agents, i.e., with $\tau \leq \frac12$, and $k>2$ types IRD convergence is no longer guaranteed. 

For the $1$-$1$-variant an agent compares the number of neighboring agents of her type with the size of the largest group of agents with a different type in her neighborhood. This captures the realistic setting where agents simply try to avoid being in a neighborhood where another group of agents dominates. We will show that even on a $\Delta$-regular network an improving response cycle exists for the $1$-$1$-SSG for sufficiently high~$\tau$.

\subsection{IRD Convergence for the One-versus-All Version} \label{sec:dynamictwotypes}

For SSGs with $k=2$ on regular networks and arbitrary networks with $\tau \leq \frac12$ the existence of a potential function was shown before in~\cite{CLM18}.  
We show that this bound is tight, i.e., that for $\tau > \frac12$ IRD may not converge.

\begin{theorem}
	IRD are not guaranteed to converge in the SSG with $k=2$ for $\tau \in \left( \frac12, 1 \right)$ on arbitrary networks. Moreover, weak acyclicity is violated.
	\label{thm:2SSG_BRC}
\end{theorem}

\begin{proof}
	We prove the statement by providing an improving response cycle where in every step exactly one improving swap is possible. The construction is shown in Fig.~\ref{fig:SSGIRC} and we assume that $x$ is sufficiently large, e.g., $x = \max \left( \lceil  \frac{1}{\tau - 0.5} \rceil,\lceil \frac{1}{2 - 2\tau} \rceil \right)$.  
	
	We have orange agents of type $T_1$ and blue agents of type $T_2$. The orange agents in the groups $u_i$ and the blue agents in the groups $v_i$, respectively, with $1 \leq i \leq 4$ are interconnected and form a clique. 
	
	During the whole cycle the agents in $u_i$ and $v_i$, respectively, are content. An orange agent $z \in u_i$ has $4x$ neighbors and at most one neighbor is blue. Hence, the positive neighborhood ratio of agent $z$ is larger than $\tau$. The same applies for a blue agent $y \in v_i$. The agent $y$ has $4x-3$ neighbors and at most one neighbor is orange. Therefore, an agent $z \in u_i$ and an agent $y \in v_i$, respectively, never have an incentive to swap their position with another agent, since they are content.
	
	In the initial placement (Fig.~\ref{fig:SSGIRC}(a)), both agents $a$ and $d$ are discontent. By swapping their positions, agent $a$ can decrease her cost from $\tau - \frac{1}{3}$ to $\tau - \frac{x-1}{2x}$ and agent $d$ decreases her cost from $\tau - \frac{x+1}{2x}$ to $\max\left(0, \tau - \frac{2}{3}\right)$. This is the only possible swap since neither $b$ nor $c$ have the opportunity to improve their costs via swapping with $c$, $d$, and $a$, $b$, respectively.
	However, after the first swap (Fig.~\ref{fig:SSGIRC}(b)) agent $a$ is still not content. Swapping with agent $c$ decreases agent $a$'s cost to $\tau - \frac{2x-1}{4x}$, and agent $c$ can decrease her cost from $\tau - \frac{2x+1}{4x}$ to $\tau - \frac{x+1}{2x}$. Again, no other swap is possible since agent $b$ would increase her cost by swapping with agent $c$ or $d$. After this (Fig.~\ref{fig:SSGIRC}(c)), agent $b$ and $d$ have the opportunity to swap and decrease their cost from $\tau - \frac{x+1}{2x}$ to $\max\left(0, \tau - \frac{2}{3}\right)$ and $\tau - \frac{1}{3}$ to $\tau - \frac{x-1}{2x}$, respectively. Once more there is no other valid swap. Agent $a$ does not want to swap with agent $d$ and agent $b$ not with agent $c$. Finally (Fig.~\ref{fig:SSGIRC}(d)), agent $a$ and $d$ swap and both agents decrease their costs to $\tau - \frac{1}{2}$. Neither does agent $b$ want to swap with agent $c$ nor can agent $c$ improve by swapping with agent $a$.
	After the fourth step the obtained placement is equivalent to the initial placement (Fig.~\ref{fig:SSGIRC}(a)), only the blue agents $a$ and $b$, and the orange agents $c$ and $d$, respectively, have exchanged positions. 
	
	Since all the executed swaps were the only possible strategy changes, this proves that the SSG is not weakly acyclic, since, starting with the given initial placement, there is no possibility to reach a stable placement via improving swaps.
	
\begin{figure}[t!]
	\hspace*{0.2cm}
	\begin{subfigure}{.23\textwidth}
		\centering
		\includegraphics[width=.9\linewidth]{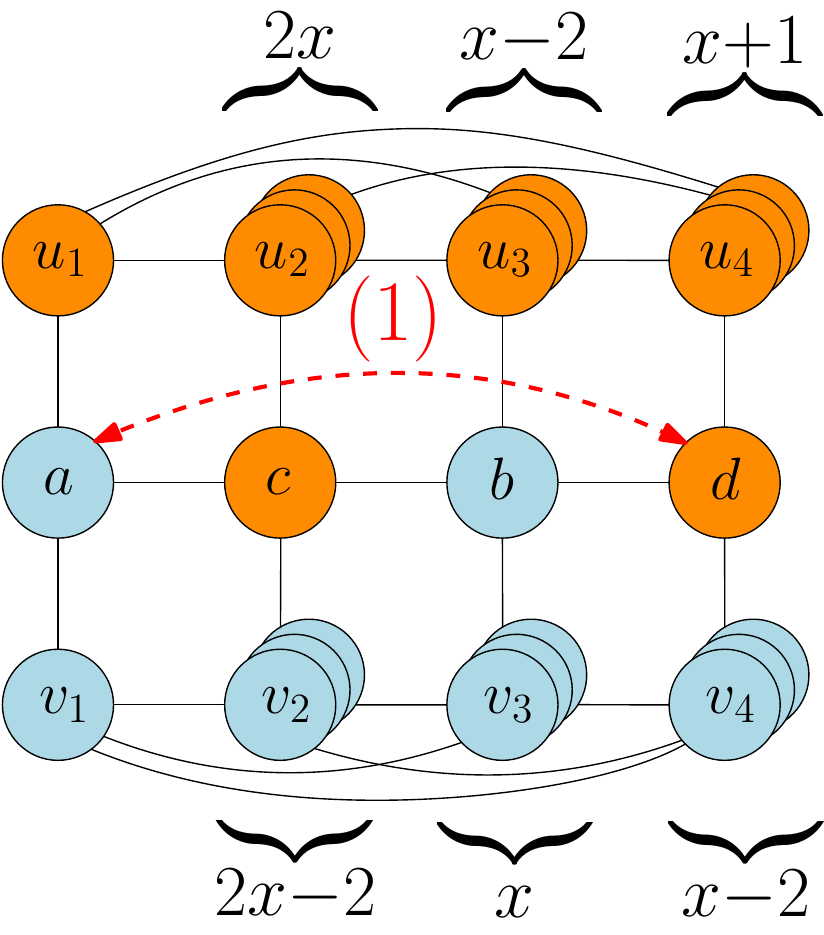}
		\caption{Initial placement\\~}
		\label{SSGIRC:1}
	\end{subfigure}
	~
	\begin{subfigure}{.23\textwidth}
		\centering
		\includegraphics[width=.9\linewidth]{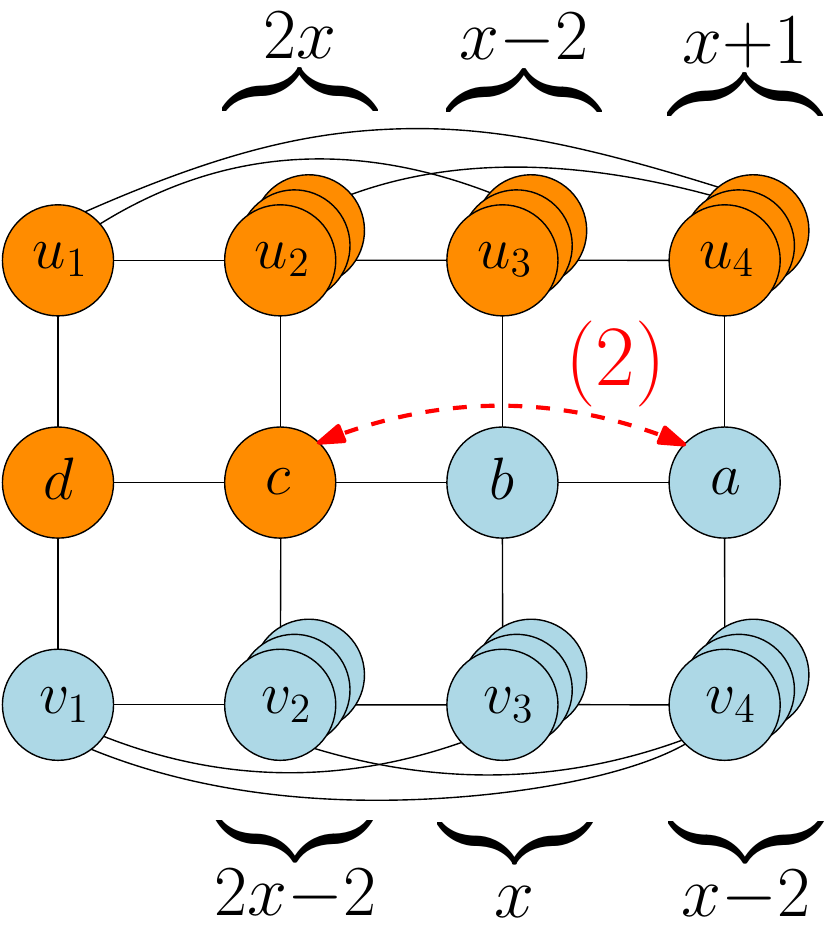}
		\caption{Placement after \\ the first swap}
		\label{SSGIRC:2}
	\end{subfigure}
	~
	\begin{subfigure}{.23\textwidth}
		\centering
		\includegraphics[width=.9\linewidth]{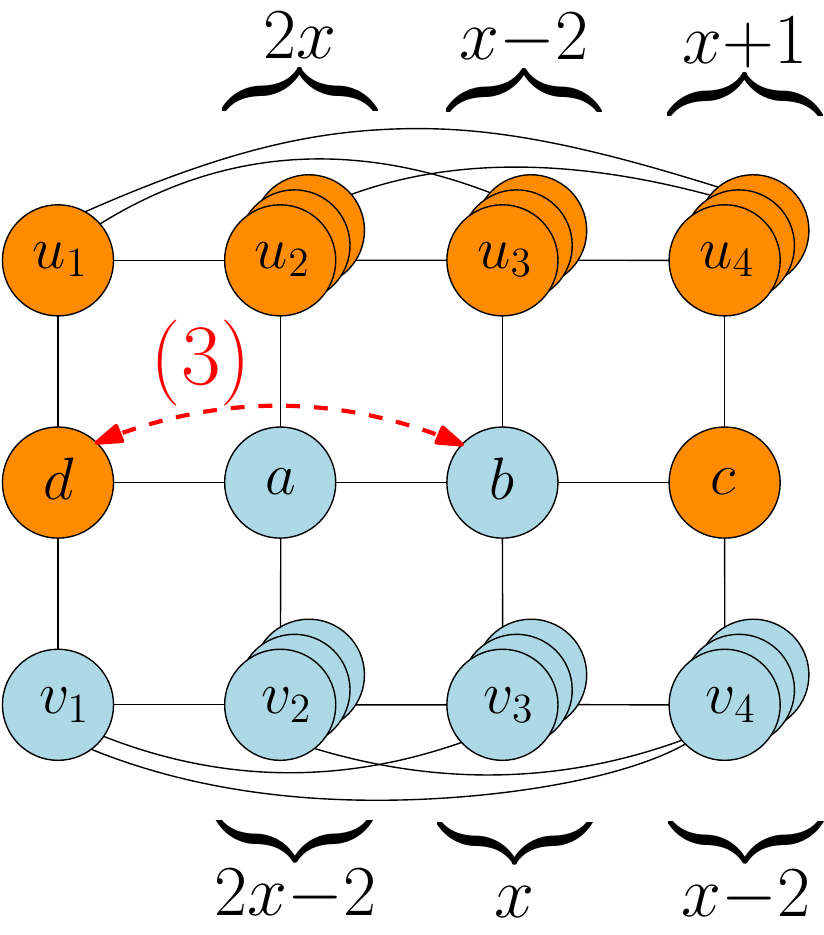}
		\caption{Placement after \\ the second swap}
		\label{SSGIRC:3}
	\end{subfigure}
	~
	\begin{subfigure}{.23\textwidth}
		\centering
		\includegraphics[width=.9\linewidth]{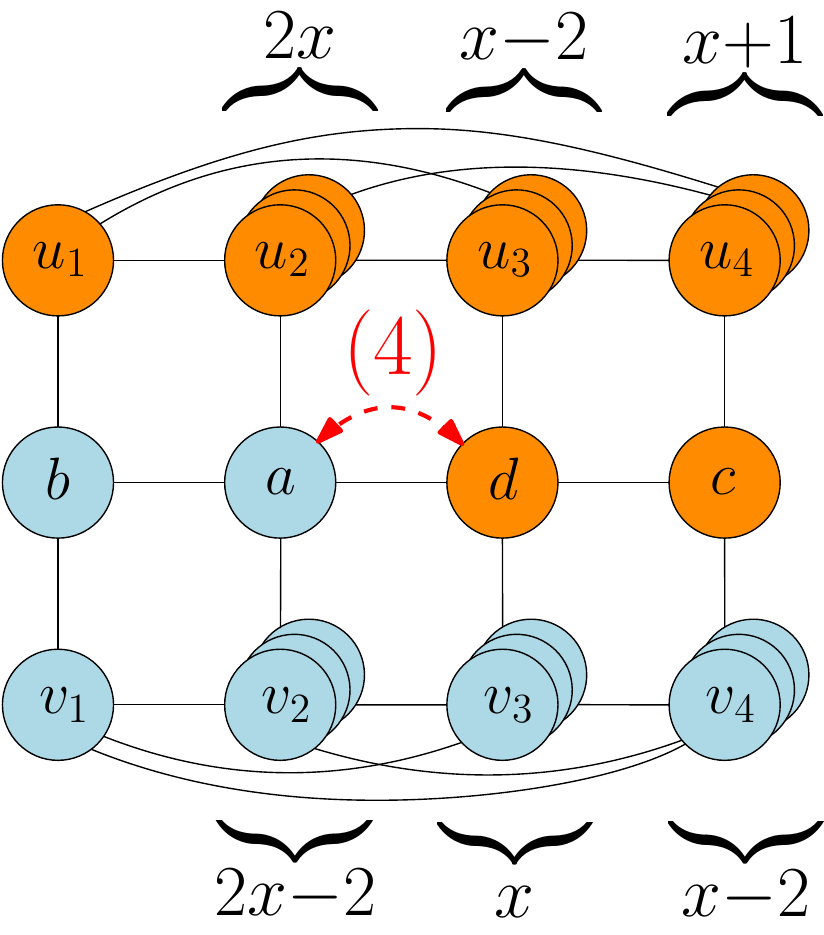}
		\caption{Placement after \\the third swap}
		\label{SSGIRC:4}
	\end{subfigure}
	\caption{An IRC for the SSG with $x = \max \left( \lceil  \frac{1}{\tau - 0.5} \rceil,\lceil \frac{1}{2 - 2\tau} \rceil \right)$ for $\tau \in \left( \frac12, 1 \right)$. The agents types are marked orange and blue. Multiple nodes in series represent a clique of nodes of the stated size. Edges between cliques or between a clique and single nodes represent that all involved nodes are completely interconnected.}
	\label{fig:SSGIRC}
\end{figure} 
\end{proof}

\noindent We now generalize the results from \cite{CLM18} by showing that convergence is guaranteed for the $1$-$k$-SSG for any $k \geq 2$.

\begin{theorem}
	IRD are guaranteed to converge in $\mathcal{O}(|E|)$ moves for the $1$-$k$-SSG with $\tau \in (0,1)$ on any $\Delta$-regular network $G=(V,E)$.\label{thm:1nSSG_reg}
\end{theorem}

\begin{proof}

	We show that $\Phi(\pg) = \frac{1}{2} \sum_{a \in A} |\npg^-(a)|$ is an ordinal potential function.
 An agent $a$ has no incentive to swap if she is content and she will never swap with an agent who has her own type, since this cannot be an improvement for both agents. Therefore, there will only be swaps between discontent agents of different types.
	Since we consider a $\Delta$-regular network we have $|\npg(a)| = |\npg^+(a)| + |\npg^-(a)| = \Delta$ for all $a \in A$.
	
	A swap between two agents $a$ and $b$ changes the current placement $\pg$ only in the locations of the involved agents and yields a new placement ${\pg'}$.
	Since a swap is an improvement for the agent $a$ who swaps, it holds that
	\[
	\frac{|\npg^+(a)|}{\Delta} < \frac{|N_{\pg'}^+(a)|}{\Delta}.
	\]
	The same is true for the other agent $b$.
	Thus the following holds for agent $a$ (and agent $b$ likewise)
	$$ 
	|\npg^+(a)|  <  |N_{\pg'}^+(a)| \iff
	\Delta - |\npg^-(a)| <  \Delta - |N_{\pg'}^-(a)| \iff
	|N_{\pg'}^-(a)|  < |\npg^-(a)|.
	$$
	It follows that $\Phi(\pg) - \Phi(\pg') > 0$ and therefore the potential function value decreases if two agents swap their current position.
	
	Since $\Phi(\pg) \leq m$ where $m$ is the number of edges in the underlying network and $\Phi(\pg)$ decreases after every swap by at least $1$ the IRD find an equilibrium in $\mathcal{O}(m)$.
\end{proof}

\noindent We contrast the above result by showing that guaranteed IRD convergence is impossible for any $\tau$ on  arbitrary networks. This emphasizes the influence of the number of agent types on the convergence behavior of the IRD.
\begin{theorem}
	IRD are not guaranteed to converge in the $1$-$k$-SSG with $k > 2$ for $\tau \in (0,1)$ on arbitrary networks. Moreover, weak acyclicity is violated. 
	\label{thm:1nSSG_BRC}
\end{theorem}

\begin{proof}	
	We give an example of an improving response cycle, where in every step exactly one improving swap exists, for any $\tau \leq 0.5$. Together with the improving response cycle given in Theorem~\ref{thm:2SSG_BRC} for $\tau > 0.5$ this yields the statement.
		
	Consider Fig.~\ref{fig:OneonAllBRC} with a sufficiently high $x$, e.g., $x > \frac{3}{4\tau} - 1$ and  $\tau \leq 0.5$. 
	We have orange agents of type $T_1$, blue agents of type $T_2$ and gray agents of type $T_3$. The agents in one group $u_i$ and $v_j$, respectively, with $1 \leq i \leq 4$ and $1\leq j \leq 2$ are interconnected and form a clique.
	
	During the whole cycle the agents in $u_i$ and $v_j$, respectively, are content. An agent in $u_i \cup v_j$ has at most two neighboring agents of different type and at least two agents of her type. Since $\tau \leq 0.5$ these agents are content. Therefore they have no incentive to swap.
	In the initial placement (Fig.~\ref{fig:OneonAllBRC}(a)), agents $a$ and $d$ are discontent and want to swap. Agent $a$ decreases her cost from $\tau$ to $\tau - \frac{1}{4(x+1)}$ while agent $d$ becomes content after the swap. This is the only possible swap. Agent $c$ does not want to swap with agent $a$ or $b$ since she increase her cost, as well agent $b$ cannot improve by swapping with $c$ or $d$.
	Then (Fig.~\ref{fig:OneonAllBRC}(b)), agent $a$ is still discontent and willing to swap her position with another agent. Swapping with agent $c$ decreases her cost to $\tau - \frac{3}{8(x+1)}$ while $c$ can improve from $\tau - \frac{5}{8(x+1)}$ to  $\tau - \frac{3}{4(x+1)}$. Again, this is the only possible swap, since $d$ is content and $c$ still doesn't want to swap with $b$.
	After this (Fig.~\ref{fig:OneonAllBRC}(c)), agent $d$ has no neighbor of her type, so she swaps with agent $b$ who becomes content. Agent $d$ reduces her cost from $\tau$ to $\tau - \frac{1}{4(x+1)}$. Agent $a$ does not want to swap with $d$ and agent $b$ not with $c$ since both $a$ and $b$ would have no agent of their own type in their neighborhood.
	Finally (Fig.~\ref{fig:OneonAllBRC}(d)), agents $a$ and $d$ want to swap. Agent $d$  decreases her cost to $\tau - \frac{1}{2(x+1)}$ and agent $a$ decreases her cost from $\tau - \frac{3}{8(x+1)}$ to $\tau - \frac{1}{2(x+1)}$. No other two agents have any incentive to swap their position, since neither agent $c$ nor $d$ want to swap with agent $b$ since they would not have a neighboring agent of their type. For the same reason agent $a$ is not interested in swapping with $c$.
	The resulting placement is equivalent to the initial one, only the blue agents $a$ and $b$ and the orange agents $c$ and $d$ exchanged positions. 
	
	Since all swaps are the only ones possible, this shows that the $1$-$k$-SSG is not weakly acyclic as there is no possibility to reach a stable placement.

\begin{figure}[t!]
	\hspace*{0.2cm}
	\begin{subfigure}{.23\textwidth}
		\centering
		\includegraphics[width=.9\linewidth]{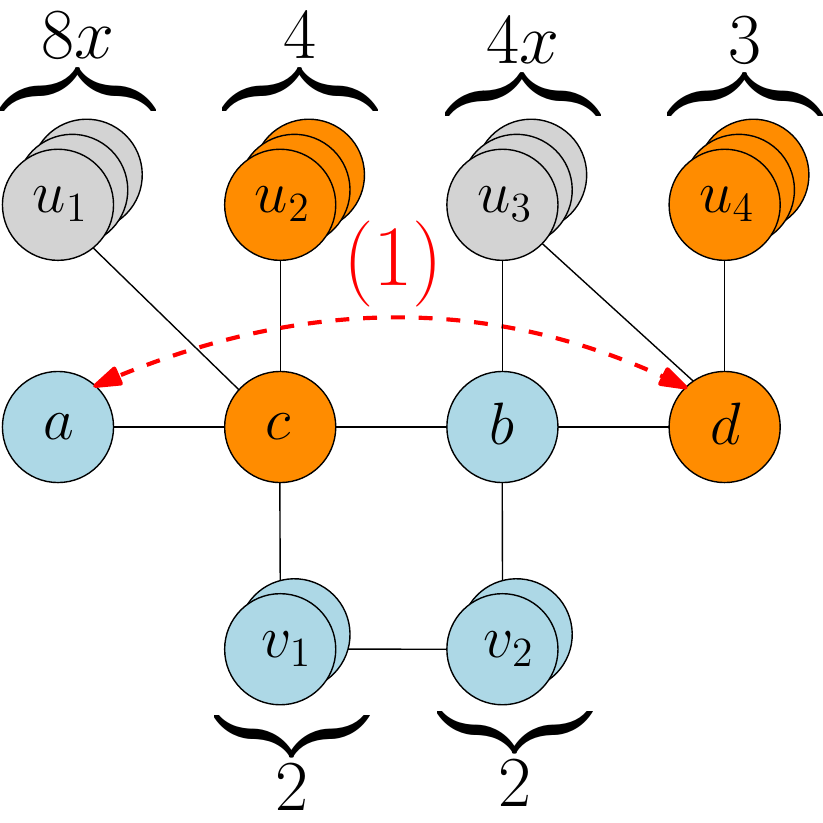}
		\caption{Initial placement\\~}
		\label{OneonAllBRC:1}
	\end{subfigure}
	~
	\begin{subfigure}{.23\textwidth}
		\centering
		\includegraphics[width=.9\linewidth]{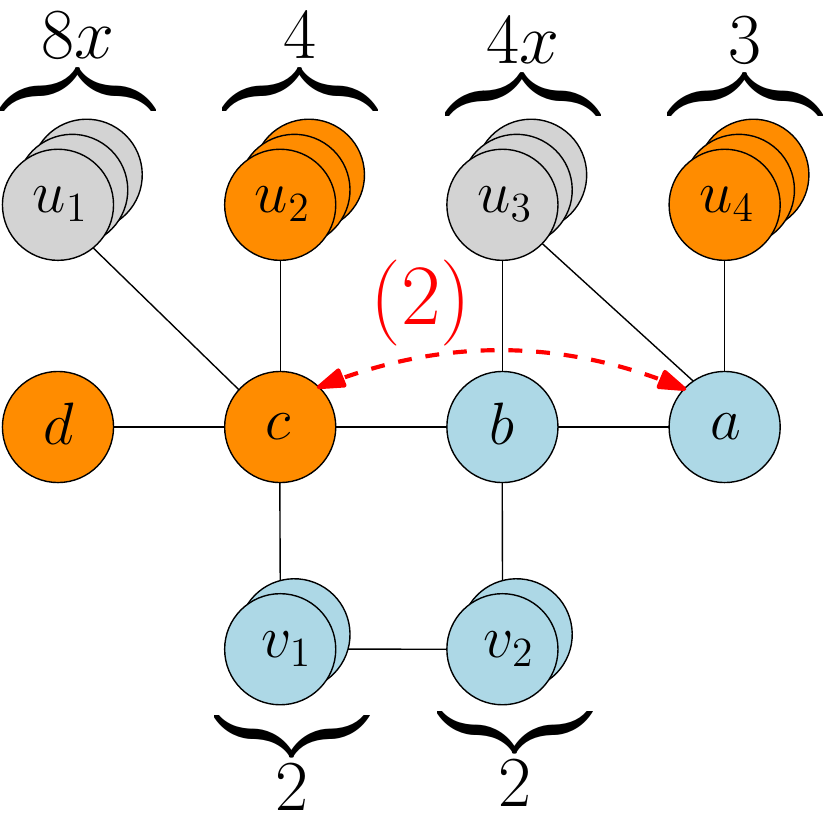}
		\caption{Placement after \\ the first swap}
		\label{OneonAllBRC:2}
	\end{subfigure}
	~
	\begin{subfigure}{.23\textwidth}
		\centering
		\includegraphics[width=.9\linewidth]{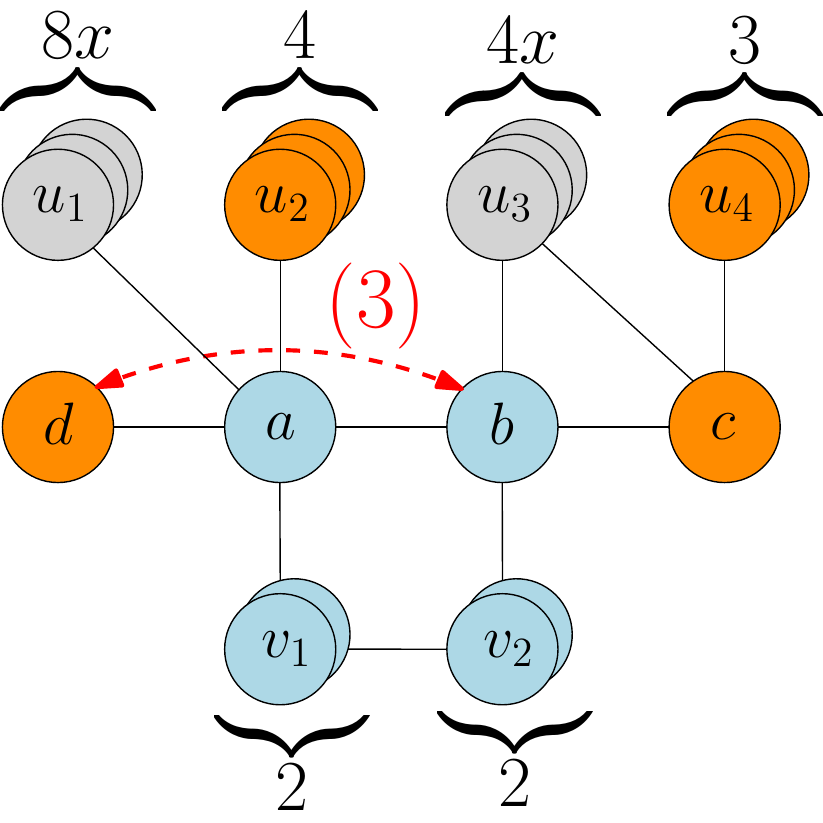}
		\caption{Placement after \\ the second swap}
		\label{OneonAllBRC:3}
	\end{subfigure}
	~
	\begin{subfigure}{.23\textwidth}
		\centering
		\includegraphics[width=.9\linewidth]{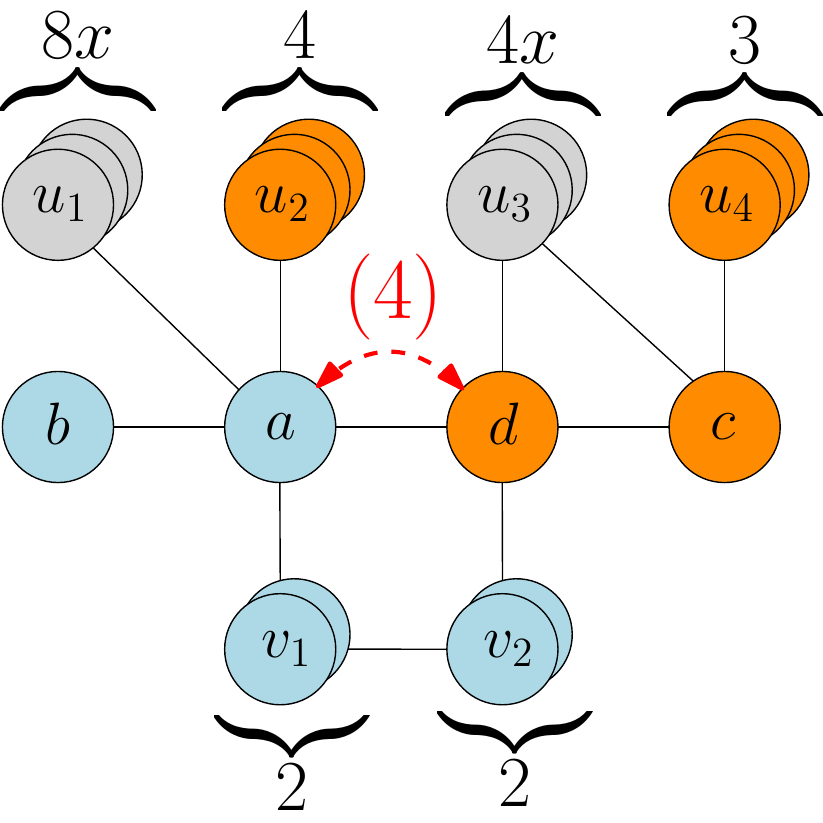}
		\caption{Placement after \\the third swap}
		\label{OneonAllBRC:4}
	\end{subfigure}
	\caption{An IRC for the $1$-$k$-SSG with $x > \frac{3}{4\tau} - 1$ for any $\tau \in (0,0.5]$. Agent types are marked orange, blue and gray. Multiple nodes in series represent a clique of nodes of the stated size. Edges between cliques or between a clique and single nodes represent that all involved nodes are completely interconnected.} 
	\label{fig:OneonAllBRC}
\end{figure}
\end{proof}

\subsection{IRD Convergence for the One-versus-One Version} \label{sec:dynamictwotypes2}

Remember, that in the $1$-$1$-SSG and $1$-$1$-JSG, respectively, an agent only considers the largest group of neighboring agents of one type, which is different from her own type. 
We start with a simple positive result for the $1$-$1$-SSG.
\begin{theorem}
	IRD are guaranteed to converge in $\mathcal{O}(|A|)$ moves, where $A$ is the set of agents, for the $1$-$1$-SSG with $\tau \leq \frac{1}{\Delta}$ on any $\Delta$-regular network $G = (V,E)$.
	\label{thm:11SSG-convergence}
\end{theorem}\begin{proof}
Any agent $a$ of type $T$ who has a neighbor $b$ of the same type is content, since $\tau \leq \frac1{\Delta}$. Since $b$ has $a$ as a neighbor, $b$ will also be content. Since both agents are content, neither of them will consider to swap positions, and therefore both will remain content.

Any agent $a$ who is discontent can't have a neighbor of the same type, otherwise $a$ would be content. The cost of $a$ must be $\tau$ in this case. Since $a$ only considers a swap that decreases her cost, after swapping the cost of $a$ can be at most $\max(0, \tau - \frac1{\Delta})$, which means $a$ is content and will continue to be so, as we showed before.

Since agents are content at least after their first swap, and agents that are content will never swap again, each agent will participate in at most one swap. Therefore, the game converges after at most $|A|$ swaps.

\end{proof}

\noindent If $\tau$ is high enough, then the $1$-$1$-SSG is no longer a potential game.  

\begin{theorem}
	IRD are not guaranteed to converge in the $1$-$1$-SSG for  $\tau \geq \frac{6}{\Delta}$ on $\Delta$-regular networks.
	\label{11SSG-IRC-reg}
\end{theorem}

\begin{proof}
	
	We use a similar instance as in the proof of Theorem~\ref{thm:11SSG-BRC}. Consider Fig.~\ref{fig:OneonOneBRC} with $x > \frac{5(1-\tau)}{6\tau}$. 
	We omit the edges between the cliques $u_1$, $u_2$ and $u_3$, of gray agents. Now, the highest degree in the graph is $6(x+1)$. In order to make the graph regular, we insert new nodes filled with agents such that each new agent is the only agent of its type, and connect these new nodes with existing nodes and each other as needed. 
	
	In the initial placement (Fig.~\ref{fig:OneonOneBRC}(a)) agent $a$ and $d$ are discontent and want to swap. Agent $a$ decreases her cost from $\tau$ to $\tau - \frac{1}{3x+1}$  while agent $d$ becomes either content after the swap or, if $\tau > \frac{1}{2}$, has costs of $\tau - \frac{1}{2}$. Then (Fig.~\ref{fig:OneonOneBRC}(b)), 
	agent $a$ is still discontent. Swapping with agent $c$ decreases her cost to $\tau - \frac{2}{4x+2}$ while agent $c$ can improve from $\tau - \frac{2}{4x+2}$ to $\tau - \frac{2}{3x+2}$.
	In the next step (Fig.~\ref{fig:OneonOneBRC}(c)), agent $d$ has no neighboring agent of her type. Therefore she swaps with agent $b$ who becomes content, if $\tau \leq \frac12$, as a result of the swap or has costs equal $\tau - \frac{1}{2}$. Agent $d$ reduces her cost from $\tau$ to $\tau - \frac{1}{6x+1}$. 
	Finally (Fig.~\ref{fig:OneonOneBRC}(d)) agent $a$ and agent $d$ want to swap. Agent $d$ has the possibility to decrease her cost to $\tau - \frac{1}{4x+1}$ and agent $a$ can decrease her own cost from $\tau - \frac{3}{4x+3}$ to $\tau - \frac{5}{6x+5}$.

	From $x>\frac{5(1-\tau)}{6\tau}$ as our only limitation and $\Delta = 6(x+1)$ we obtain $\tau \geq \frac{6}{\Delta}$, where equality is reached if $x$ is chosen as low as possible.
\end{proof}
The situation is much worse on arbitrary graphs as the following theorem shows.
\begin{theorem}
	IRD are not guaranteed to converge in the $1$-$1$-SSG for $\tau \in (0,1)$ on arbitrary networks. Moreover, weak acyclicity is violated.
	\label{thm:11SSG-BRC}
\end{theorem}

\begin{proof}
	
	We show the statement by giving an example for an improving response cycle where in every step exactly one improving swap exists.
	Consider Fig.~\ref{fig:OneonOneBRC} with $x > \max \left( \frac{5(1-\tau)}{6\tau}, \frac{\tau}{1-\tau} \right)$. 
	We have orange agents of type $T_1$, blue agents of type $T_2$ and gray agents of type $T_3$. The agents in one group $u_i$ and $v_i$, respectively, with $i \in \{1,2,3,4,5\}$ are interconnected and form a clique.
	
During the whole cycle the agents in $u_i$ and $v_i$, respectively, are content. Agent $v_2$ has at most $2$ neighbors of any type other than $T_1$ and at least $3x$ neighbors of her own type. All the other agents in $u_i \cup v_i$ have at most one neighbor of another type and at least $x$ neighboring agents of their own type. Therefore the positive neighborhood ratio \textit{pnr} of an agent $z \in u_i \cup v_i$ is larger than $\tau$ for $x > 1$ and $z$ has no incentive to swap.
	In the initial placement (Fig.~\ref{fig:OneonOneBRC}(a)) agent $a$ and $d$ are discontent and want to swap. Agent $a$ decreases her cost from $\tau$ to $\tau - \frac{1}{3x+1}$ while agent $d$ becomes content after the swap. This is the only possible swap. Agent $c$ does not want to swap with agent $a$ or $b$ since she would be worse off and agent $b$ cannot improve by swapping with $d$.
	Then (Fig.~\ref{fig:OneonOneBRC}(b)), 
	agent $a$ is still discontent. Swapping with agent $c$ decreases her cost to $\tau - \frac{2}{4x+2}$ while agent $c$ can improve from $\tau - \frac{2}{4x+2}$ to $\tau - \frac{2}{3x+2}$. Again, this is the only possible swap, since $d$ is content and $c$ would not improve by swapping with agent $b$.
	In the next step (Fig.~\ref{fig:OneonOneBRC}(c)), agent $d$ has no neighboring agent of her type. Therefore she swaps with agent $b$ who becomes content as a result of the swap. Agent $d$ reduces her cost from $\tau$ to $\tau - \frac{1}{6x+1}$. Agent $a$ does not want to swap with $d$ since at the new position she wouldn't have a neighboring agent of her own type and agent $b$ not with $c$ since this wouldn't be an improvement for $b$.
	Finally (Fig.~\ref{fig:OneonOneBRC}(d)) agent $a$ and agent $d$ want to swap. Agent $d$ has the possibility to decrease her cost to $\tau - \frac{1}{4x+1}$ and agent $a$ can decrease her own cost from $\tau - \frac{3}{4x+3}$ to $\tau - \frac{5}{6x+5}$. No other two agents have the incentive to swap their position, since agent $c$ does not want to swap with agent $a$ or $b$. 
	We end up in a placement which is equivalent to the initial one, only the blue agents $a$ and $b$ and the orange agents $c$ and $d$ exchanged positions. 
	
	Since all swaps were the only ones possible, this shows that the $1$-$1$-SSG is not weakly acyclic as there is no possibility to reach a stable placement.
	
	\begin{figure}[t!]
	\hspace*{0.2cm}
	\begin{subfigure}{.23\textwidth}
		\centering
		\includegraphics[width=.9\linewidth]{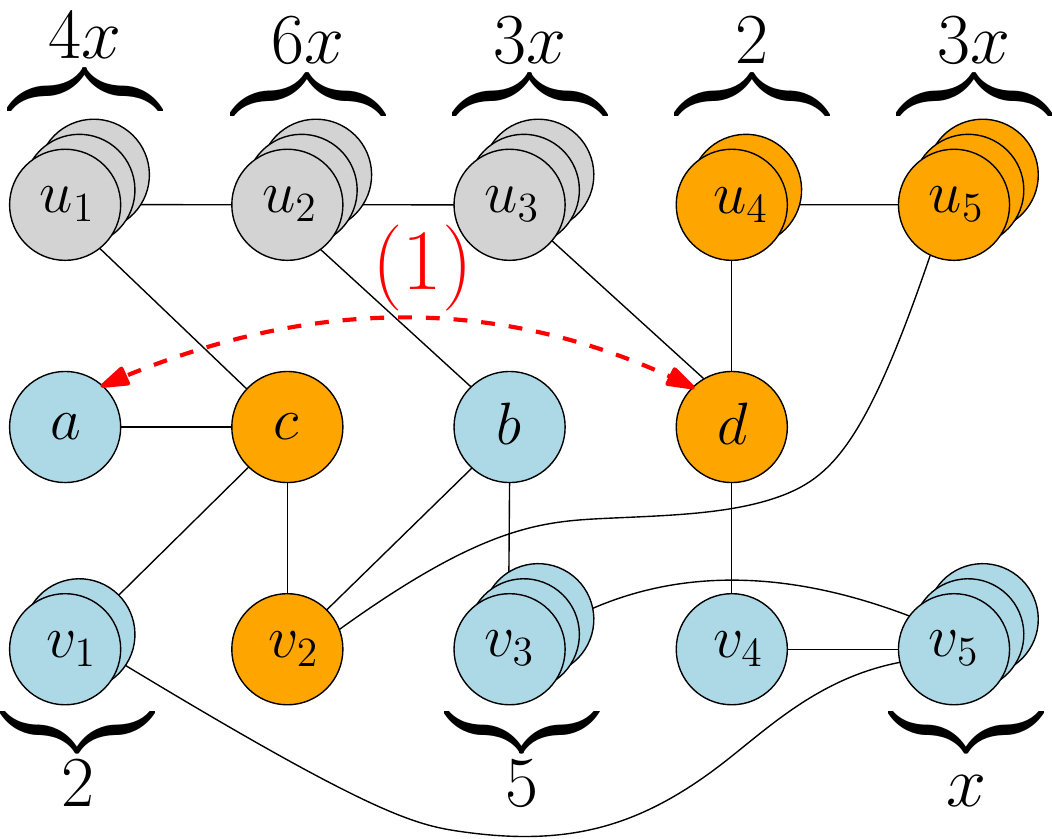}
		\caption{Initial placement\\~}
		\label{OneonOneBRC:1}
	\end{subfigure}
	~
	\begin{subfigure}{.23\textwidth}
		\centering
		\includegraphics[width=.9\linewidth]{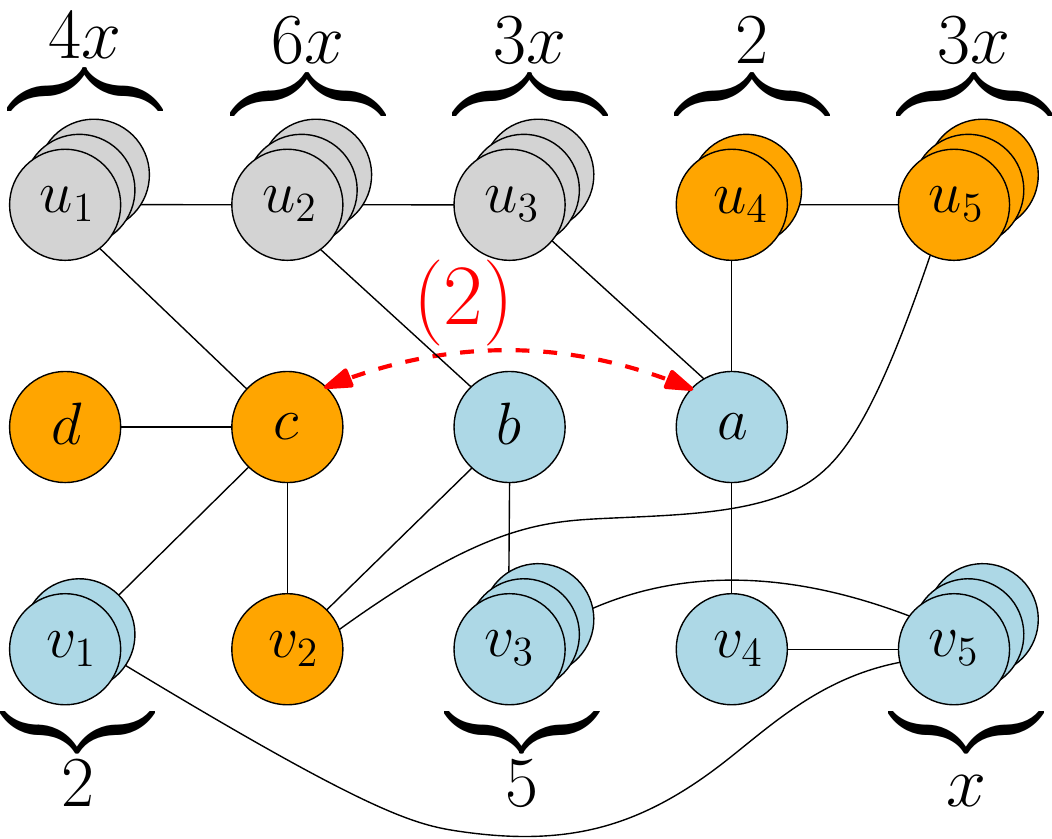}
		\caption{Placement after \\ the first swap}
		\label{OneonOneBRC:2}
	\end{subfigure}
	~
	\begin{subfigure}{.23\textwidth}
		\centering
		\includegraphics[width=.9\linewidth]{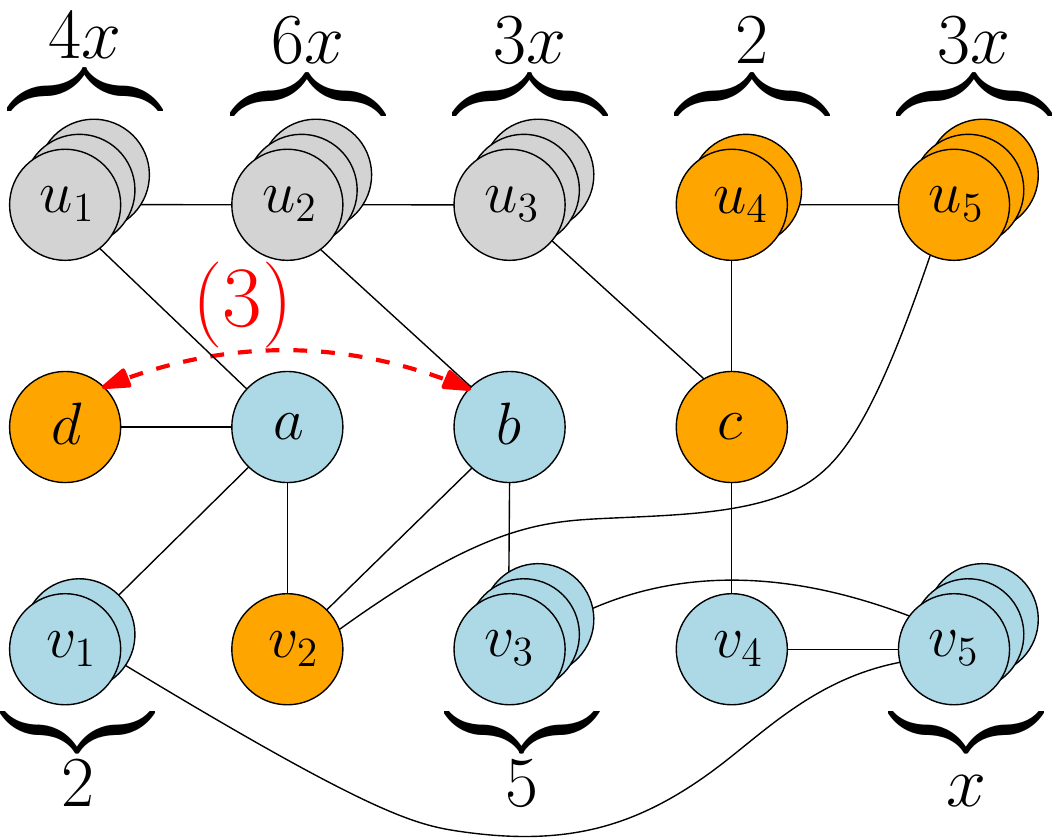}
		\caption{Placement after \\ the second swap}
		\label{OneonOneBRC:3}
	\end{subfigure}
	~
	\begin{subfigure}{.23\textwidth}
		\centering
		\includegraphics[width=.9\linewidth]{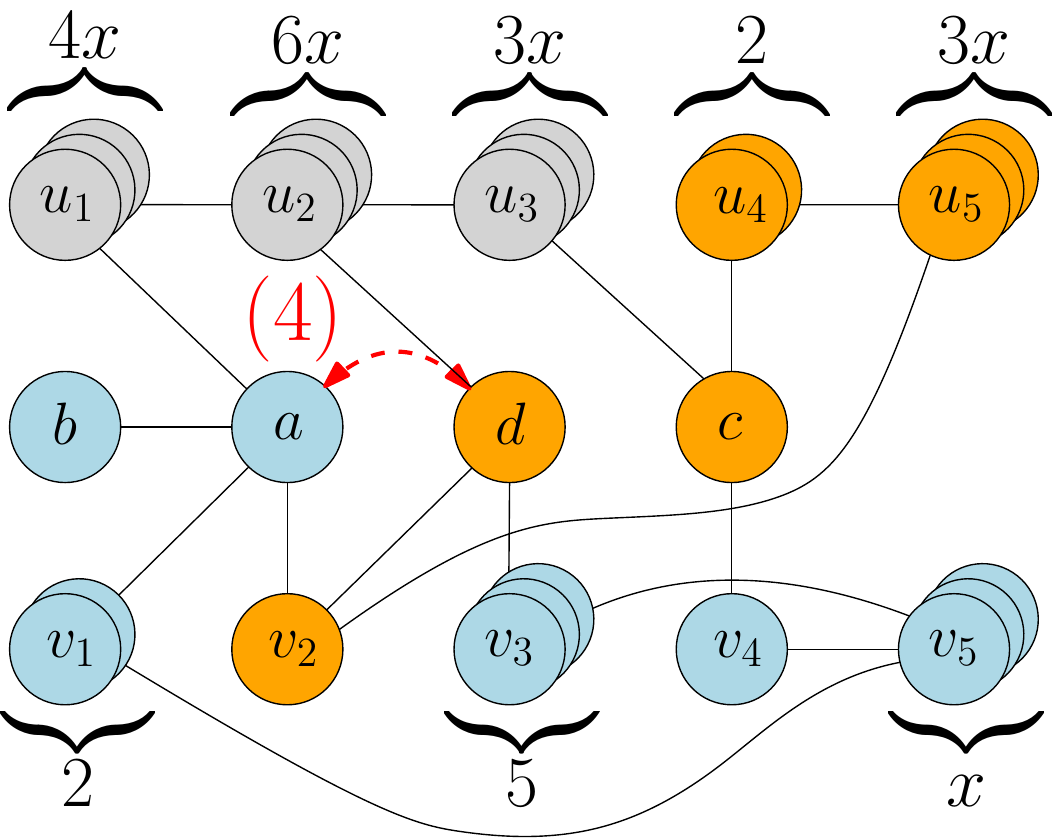}
		\caption{Placement after \\the third swap}
		\label{OneonOneBRC:4}
	\end{subfigure}
	\caption{An IRC with exactly one improving swap per step for the $1$-$1$-SSG with $x > \max \left( \frac{5(1-\tau)}{6\tau}, \frac{\tau}{1-\tau} \right)$ for any $\tau \in (0,1)$. Agents types are marked orange, blue and gray. Multiple nodes in series represent a clique of nodes of the stated size. Edges between cliques or between a clique and single nodes represent that all involved nodes are completely interconnected. 
	}
	\label{fig:OneonOneBRC}
\end{figure}
	
\end{proof}

\section{Schelling Dynamics for the Jump Schelling Game}

We now analyze the convergence behavior of IRD for the strategic segregation process via jumps. Chauhan et al.~\cite{CLM18} proved that the JSG converges for $\tau \in (0,1)$ on $2$-regular graphs. Furthermore they showed that there exists an IRC for $\tau \in \left( \frac13, \frac23 \right]$ on a $8$-regular grid if the agents have a favorite location, i.e., a node to whom an agent $a$ wants to be as close as possible without increasing her costs. In particular such a favorite location is necessary for their IRC. We show that convergence is not guaranteed even without a favorite location on arbitrary graphs and sharp the threshold for $\Delta$-regular graphs at $\tau = \frac{2}{\Delta}$.

We first turn our focus to the $1$-$k$-JSG, where an agent only distinguishes between own and other types. Hence, an agent simply compares the number of neighbors of her type with the total number of neighbors. 

\subsection{IRD Convergence for the One-versus-All Version}\label{sec:dynamictwotypesJSG}

In~\cite{CLM18} only for the JSG on $2$-regular graphs the existence of an ordinal potential function was shown. In contrast, we prove a sharp threshold result, with the threshold being at $\tau = \tfrac{2}{\Delta}$, for the convergence of IRD for the $1$-$k$-JSG on $\Delta$-regular graphs, for any $\Delta \geq 2$. Moreover, we show that the game is not weakly acyclic on arbitrary graphs.   

\begin{theorem}
	IRD are guaranteed to converge in $\mathcal{O}(|E|)$ steps for the $1$-$k$-JSG with $\tau \leq \frac{2}{\Delta}$ on any $\Delta$-regular network $G=(V,E)$.
	\label{thm:1nJSG}
\end{theorem}

\begin{proof}
	
	For any $\Delta$-regular network $G$ we define the weight $w_{\pg}(e)$ of any edge $e = \{u,v\} \in E$ as:
	$$
	w_{\pg}(e) = \begin{cases}
	1, \text{ if $u$ and $v$ are occupied by agents of different types for $\pg$},\\
	c, \text{ if either $u$ or $v$, but not both, are empty for $\pg$},\\
	0, \text{ otherwise,}
	\end{cases}
	$$
	with $\frac{1}{2} - \frac{1}{2\Delta} < c < \frac{1}{2}$.
	We prove that $\Phi(\pg) = \sum_{e \in E} w_{\pg}(e)$ is an ordinal potential function.
	
	Note that $\tau$ is sufficiently small, so that an agent becomes content if she has two neighbors of her type. Therefore, an agent who is willing to jump to another node has at most one neighbor of the same type. Without loss of generality, we assume the existence of a discontent agent $y$ for placement $\pg$. Let $\pg'$ be a placement that results from a jump of $y$.
	Let $a = |\npg^{+}(y)|$, $b = |\npg^{-}(y)|$ and let $\varepsilon$ be the number of empty nodes in the neighborhood of $\pg(y)$. Let $a' = |N_{\pg'}^{+}(y)|$ and $b' = |N_{\pg'}^{-}(y)|$ be the number of agents of the same type and of different type, respectively, and let $\varepsilon'$ be the number of empty nodes in the neighborhood of $\pg'(y)$.
	We will show that if an agent jumps, $\Phi$ changes it holds that
	\begin{align*}
	\Phi(\pg) - \Phi(\pg') & =  \left(0a + 1b + c\varepsilon + ca' + cb' + 0\varepsilon' \right) -\left(ca + cb + 0  \varepsilon + 0a' + 1b' + c\varepsilon' \right) \\
	& = -ca + (1 - c)b + c\varepsilon + ca' + (c - 1)b' - c\varepsilon' > 0,
	\end{align*}
	and therefore $\Phi$ decreases for every improving jump of an agent.
	
	There is no incentive for agent $y$ to decrease the number of neighbors of the same type because decreasing this number would mean that either $a \geq 2$, i.e., agent $y$ is content and does not want to jump, or $a = 1$ and therefore $a' = 0$ which is never an improvement.
	Hence, we have to distinguish between two cases:
	
	If $a < a'$, then agent $y$ increases the number of neighbors of the same type.
	Since we consider a $\Delta$-regular network, we have $a+b+\varepsilon = \Delta$ and  $a'+b'+\varepsilon' = \Delta$, so $b = \Delta - a - \varepsilon$ and $b' = \Delta - a' - \varepsilon'$. Hence,
	\begin{align*}
	& -ca + (1-c)b + c\varepsilon + ca' + (c-1)b' - c\varepsilon' \\
	=\ \ & -ca + (1-c)(\Delta - a - \varepsilon) + c\varepsilon + ca' + (c-1)(\Delta - a' - \varepsilon') - c\varepsilon' \\
	=\ \ & -ca + (1-c)(- a - \varepsilon) + c\varepsilon + ca' + (c-1)( - a' - \varepsilon') - c\varepsilon'  \\
	=\ \ & -ca - a - \varepsilon + ca + c\varepsilon + c\varepsilon + ca' - ca' - c\varepsilon' + a' + \varepsilon' - c\varepsilon' \\
	=\ \ & (2c-1)\varepsilon + (1-2c)\varepsilon' - a + a' \\
	\geq\ \ & (2c - 1)\varepsilon - a + a',
	\end{align*}
	since $1-2c>0$ and $\varepsilon' \geq 0$.
	If $\varepsilon = 0$, we obtain $(2c-1)\varepsilon - a + a' = -a + a' > 0$.
	If $\varepsilon > 0$, we have
	
	$$ (2c-1)\varepsilon - a + a'
	> \left(2\left(\frac{1}{2} -\frac{1}{2\Delta}\right)-1\right)\varepsilon - a + a' =\frac{-\varepsilon}{\Delta} - a + a' \geq 0,$$
	since $\frac{\varepsilon}{\Delta} \leq 1 \leq a' - a$.
	
	If $a = a'$, then the number of same type neighbors of agent $y$ stays the same.
	Since $y$ improves her positive neighborhood ratio and since $a = a'$ the number of different type neighbors of $y$ has to decrease and therefore $b' < b$. We denote the difference as $\delta$ with $b = b' + \delta$. Therefore it holds that $\delta > 0$. Since we consider a $\Delta$-regular network, it follows that $\varepsilon' = \varepsilon+\delta$. Hence,
	\begin{align*}
	& -ca + (1-c)b + c\varepsilon + ca' + (c-1)b' - c\varepsilon' \\
	=\ \ & -ca + (1-c)(b' + \delta) + c\varepsilon + ca' + (c-1)b' - c(\varepsilon+\delta) \\
	=\ \ & -ca  + (1-c)\delta + ca' - c\delta \\
	=\ \ & (1-c)\delta - c\delta\\
	=\ \ & (1-2c)\delta > 0, 
	\end{align*}
	where the second to last equality holds since $a = a'$.
	
	Since $\Phi(\pg) \leq m$ where $m$ is the number of edges in the underlying graph and $\Phi(\pg)$ decreases after every jump by at least $(1-2c)$ the IRD find an equilibrium in $\mathcal{O}(m)$.
\end{proof}

\noindent Actually Theorem~\ref{thm:1nJSG} is tight and convergence is not guaranteed if $\tau > \frac{2}{\Delta}$.  

\begin{theorem}
	The $1$-$k$-JSG for $\tau > \frac{2}{\Delta}$ on $\Delta$-regular graphs is no potential game. \label{thm:1nJSG_BRC}
\end{theorem}

\begin{proof}
We prove the statement by providing an improving response cycle. See Fig.~\ref{fig:JSGIRC}. If we have more than two types of different agents, all agents of types dissimilar from $T_1$ and $T_2$ can be placed outside of the neighborhood of the agents $a$, $b$ and $c$ who are involved in the IRC.

Let $\tau > \frac{2}{\Delta}$.
In the initial placement, agent $a$ is discontent and has cost of $\tau - \frac{2}{\Delta}$.
By jumping next to agent $c$ she becomes content.
Because of this jump, agent $b$ becomes isolated.
Jumping next to the agents $d$ and $y$ decreases her costs from $\tau$ to $\tau - \frac{1}{\Delta - 1}$.
After the second step, the obtained placement is equivalent to the initial placement.
Only agents $a$, $b$, and $c$ changed their roles.
Hence, the next two jumps from agents $c$ and $a$ are like the first two: First, agent $c$ jumps next to agent $b$ to become content, then agent $a$ jumps next to the agents $c$ and $z$ to avoid an isolated position.
We end up in an equivalent placement to the initial one.

\begin{figure}[t!]
\hspace*{0.2cm}
\begin{subfigure}{.23\textwidth}
	\centering
	\includegraphics[width=.9\linewidth]{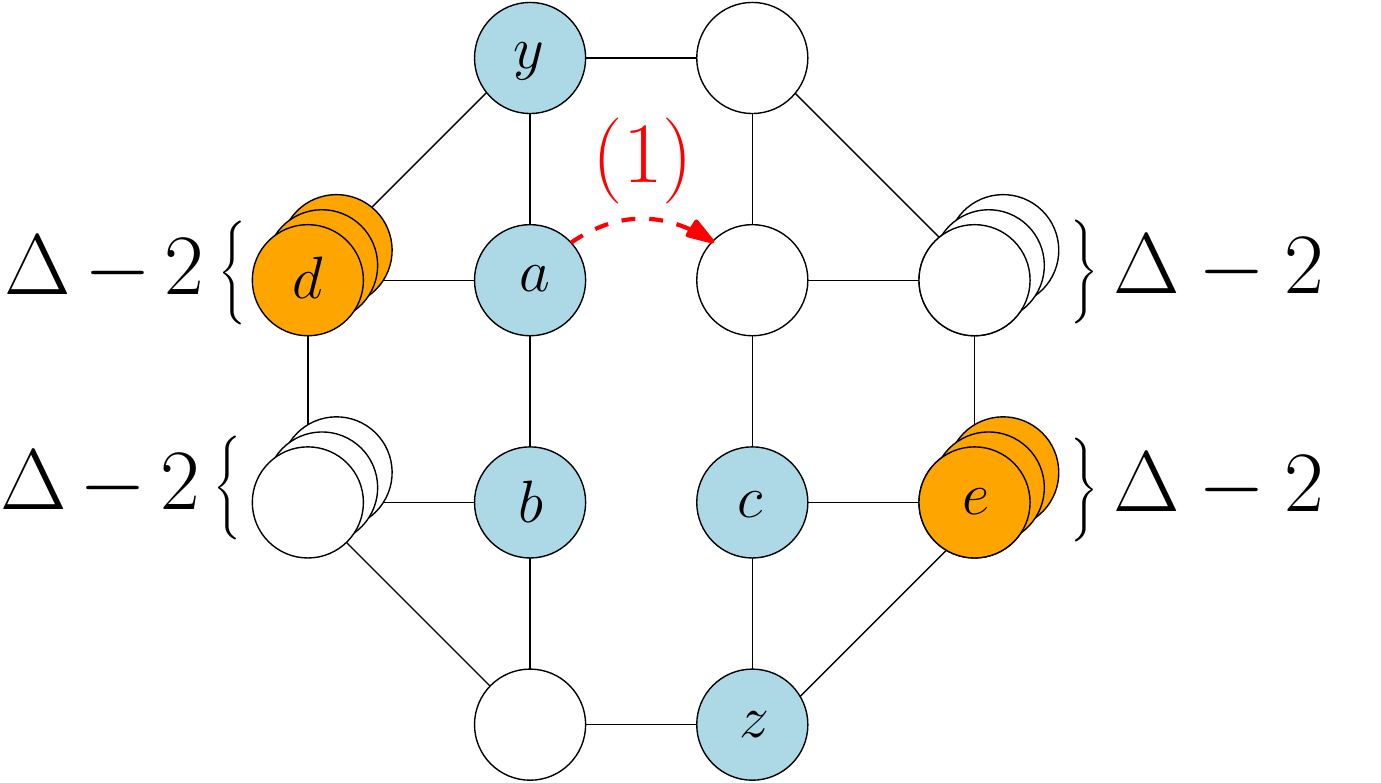}
	\caption{Initial placement\\~}
	\label{JSGIRC:1}
\end{subfigure}
~
\begin{subfigure}{.23\textwidth}
	\centering
	\includegraphics[width=.9\linewidth]{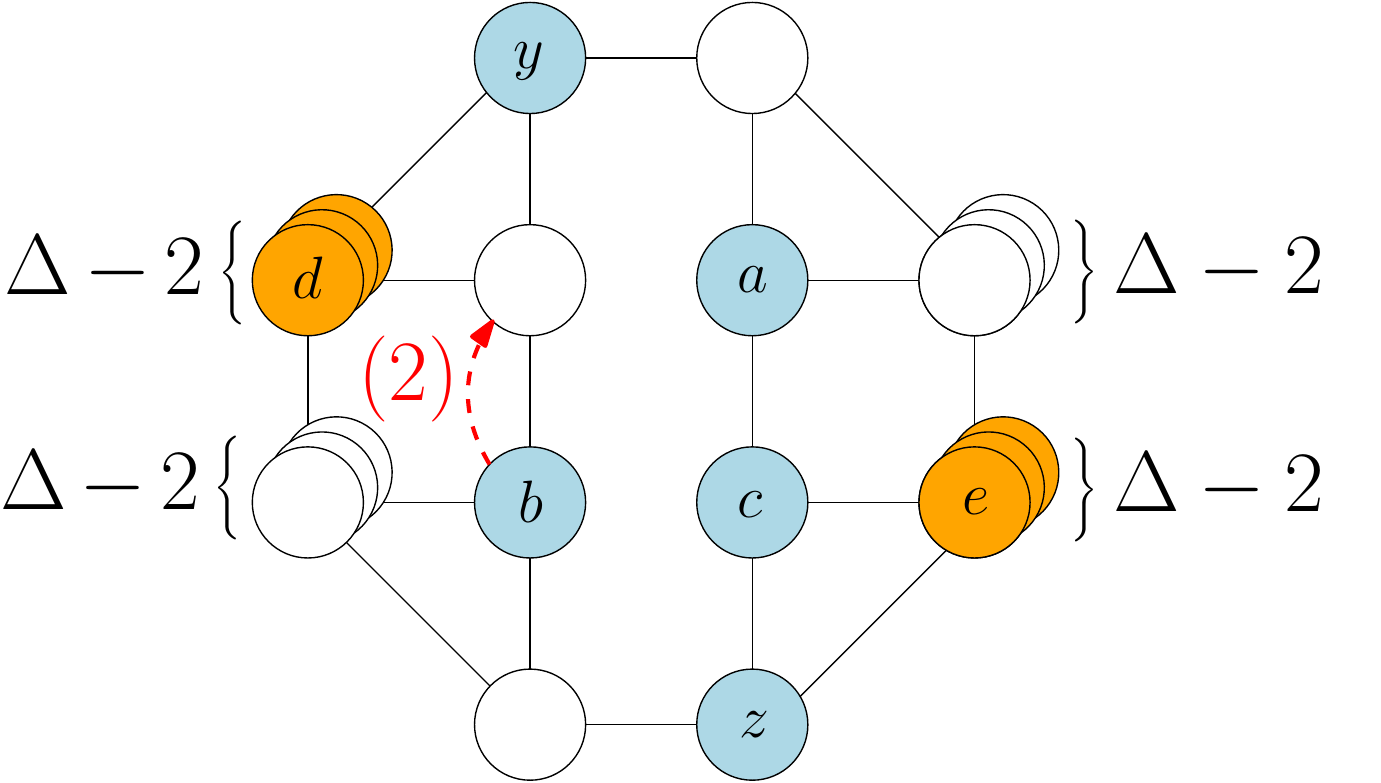}
	\caption{Placement after \\ the first swap}
	\label{JSGIRC:2}
\end{subfigure}
~
\begin{subfigure}{.23\textwidth}
	\centering
	\includegraphics[width=.9\linewidth]{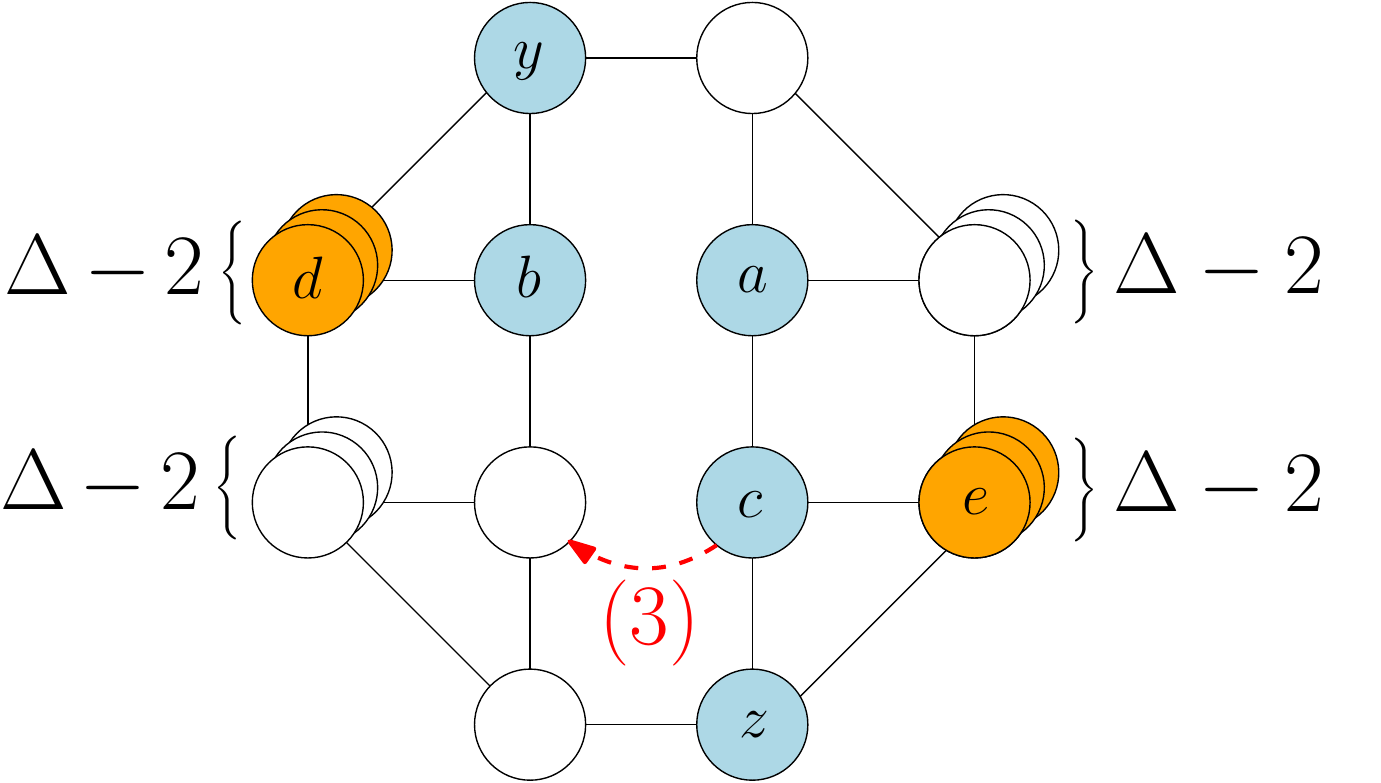}
	\caption{Placement after \\ the second swap}
	\label{JSGIRC:3}
\end{subfigure}
~
\begin{subfigure}{.23\textwidth}
	\centering
	\includegraphics[width=.9\linewidth]{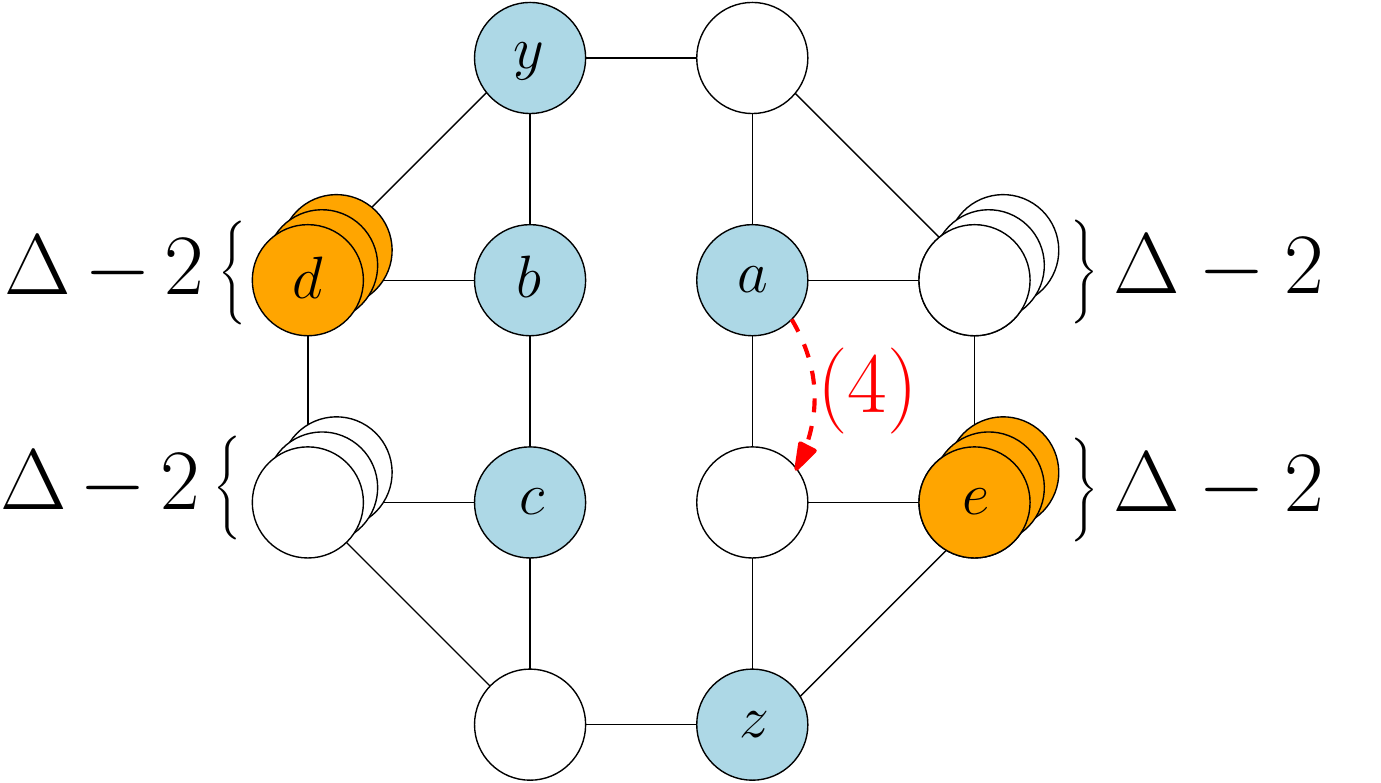}
	\caption{Placement after \\the third swap}
	\label{JSGIRC:4}
\end{subfigure}
\caption{An IRC for the JSG for $\tau > \frac{2}{\Delta}$ on a $\Delta$-regular network. Empty nodes are white, agents of type $T_1$ are orange, type $T_2$ agents are blue. Multiple nodes in series represent a clique of $\Delta-2$ nodes. An edge between a clique and a single node denotes that each clique node is connected to that single node. An edge between two cliques represents that each clique node as exactly one neighbor in the other clique. With this the network is indeed $\Delta$-regular: Each node is connected to all nodes of exactly one group of size $\Delta-2$ and to two other nodes.}
\label{fig:JSGIRC}
\end{figure}
\end{proof}

\noindent If the underlying network is an arbitrary network the situation is worse.

\begin{theorem}
	IRD are not guaranteed to converge in the $1$-$k$-JSG for $\tau \in (0,1)$ on arbitrary networks. Moreover, weak acyclicity is violated.
	\label{thm:1nJSG-IRC}
\end{theorem}

\begin{proof}
We show the statement by giving an example of an improving response cycle where in every step exactly one agent has exactly one improving jump.
Consider Fig.~\ref{fig:JSGBRC}. We assume that $x$ is sufficiently high, e.g. $x>\max \left( \frac{2}{\tau}, \frac{1}{1-\tau} \right)$. 
If we have more than two different types of agents, all agents of types dissimilar to $T_1$ and $T_2$ can be placed in cliques outside of the neighborhood of all of the agents involved in the IRC. If these cliques are placed inside network components which are neither connected to the IRC nodes, nor to each other, the agents of these types will never become discontent. Hence, the jumps of the given IRC are the only ones possible.

\begin{figure}[t!]
	\hspace*{0.2cm}
	\begin{subfigure}{.23\textwidth}
		\centering
		\includegraphics[width=.7\linewidth]{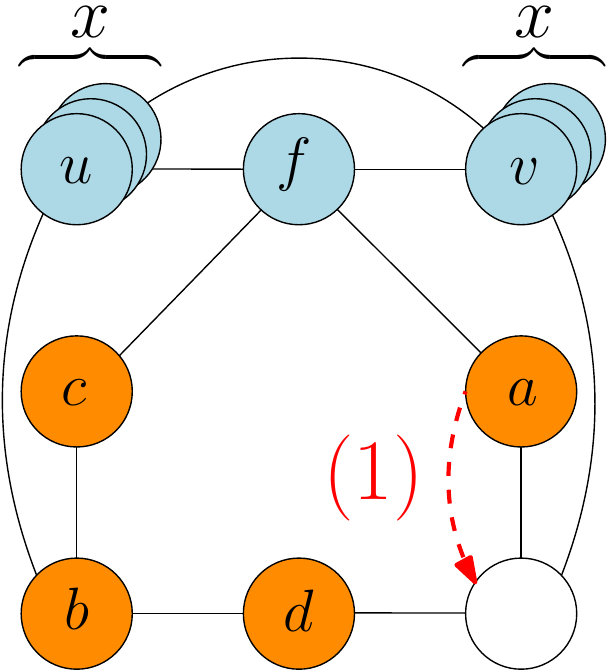}
		\caption{Initial placement\\~}
		\label{JSGBRC:1}
	\end{subfigure}
	~
	\begin{subfigure}{.23\textwidth}
		\centering
		\includegraphics[width=.7\linewidth]{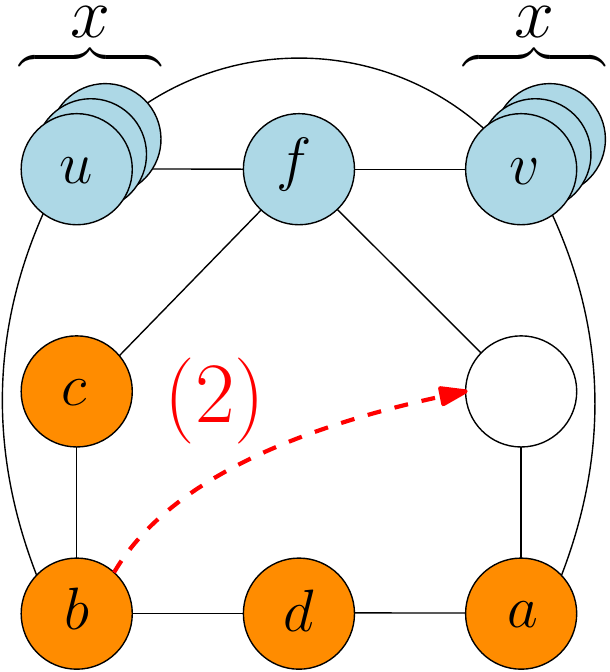}
		\caption{Placement after \\ the first swap}
		\label{JSGBRC:2}
	\end{subfigure}
	~
	\begin{subfigure}{.23\textwidth}
		\centering
		\includegraphics[width=.7\linewidth]{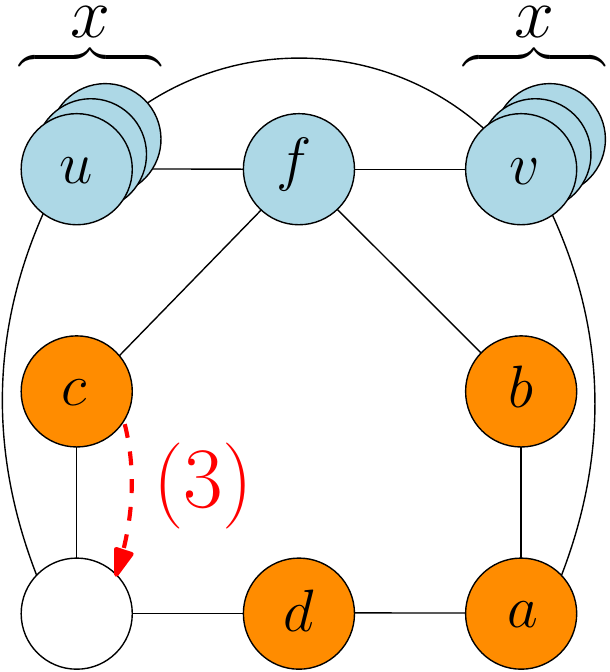}
		\caption{Placement after \\ the second swap}
		\label{JSGBRC:3}
	\end{subfigure}
	~
	\begin{subfigure}{.23\textwidth}
		\centering
		\includegraphics[width=.7\linewidth]{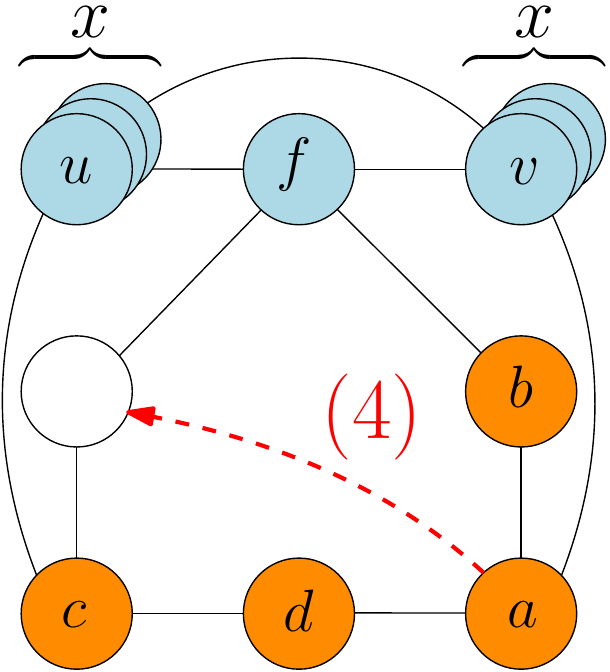}
		\caption{Placement after \\the third swap}
		\label{JSGBRC:4}
	\end{subfigure}
	\caption{An IRC with exactly one improving jump per step for the JSG for $x>\max \left( \frac{2}{\tau}, \frac{1}{1-\tau} \right)$ for any $\tau \in (0,1)$. Agents of type $T_1$ are orange, type $T_2$ agents are blue. Multiple nodes in a series represent a clique of nodes of the stated size. Edges between cliques or between a clique and single nodes represent that all involved nodes are completely interconnected.}
	\label{fig:JSGBRC}
\end{figure}

\noindent In the construction we have four orange agents, $a$, $b$, $c$, $d$, of type $T_1$ and $2x + 1$ blue agents in the sets $u$ and $v$ and $f$ of type $T_2$ and one white empty node. All nodes which are occupied by the blue agents are interconnected and form a clique.

During the whole cycle, all blue agents are content. A blue agent $z \in T_2$ has $2x+2$ neighbors of whom at least $2x$ are of the same type. Hence, the positive neighborhood ratio of an agent $z$ is larger than $\tau$ and she has no incentive to jump to another currently empty node. Also the orange agent $d$ remains content during the entire cycle since she is never isolated and has never a neighboring agent of a different type. In the initial placement (Fig.~\ref{fig:JSGBRC}(a)), the orange agent $a$ is discontent, since her only neighboring agent $f$ is blue. Therefore, $a$ jumps to the empty node. Agent $b$ and, depending on the value of $\tau$, agent $c$ are discontent. However, jumping to the empty node next to agent $d$ is not an improvement for them. Now (Fig.~\ref{fig:JSGBRC}(b)) agent $b$ is discontent, since $x$ is chosen sufficiently high that the positive neighborhood ratio of $b$ is smaller than $\tau$. Hence, jumping to the empty node next to agent $a$ improves the cost of $b$ from $\tau - \frac{2}{x+2}$ to $\max\left(0, \tau  - 0.5 \right)$. Again, this is the only valid jump, since agent $c$ would still have exactly one blue agent and one orange agent in her neighborhood by jumping next to agent $a$. After two further jumps (Fig.~\ref{fig:JSGBRC}(c) and \ref{fig:JSGBRC} (d)) by agents $c$ and~$a$, which are equivalent to those shown in Fig.~\ref{fig:JSGBRC}(a) and Fig.~\ref{fig:JSGBRC}(b), restore the initial placement.

Since all executed jumps were the only ones possible, this shows that the JSG is not weakly acyclic as there is no possibility to reach a stable placement via improving jumps.

\end{proof}

\subsection{IRD Convergence for the One-versus-One Version}

\noindent Now we turn to the $1$-$1$-JSG. By using the same proof as in Theorem~\ref{thm:11SSG-convergence} with jumps instead of swaps we get the following positive result.
\begin{theorem}
	IRDs are guaranteed to converge in $\mathcal{O}(|A|)$ moves for the $1$-$1$-JSG with $\tau \leq \frac{1}{\Delta}$ on $\Delta$-regular networks.
	\label{thm:11JSG-convergence}
\end{theorem}

\noindent The same IRC which proves Theorem~\ref{thm:1nJSG_BRC} for the $1$-$k$-JSG yields the next result. 

\begin{theorem}
	IRD may not converge in the $1$-$1$-JSG for $\tau > \frac{2}{\Delta}$ on $\Delta$-regular graphs. \label{thm:11JSG_reg_IRC}
\end{theorem}

\noindent Finally the proof of Theorem~\ref{thm:1nJSG-IRC} works for the following result as well.

\begin{theorem}
	IRD are not guaranteed to converge in the $1$-$1$-JSG for $\tau \in (0,1)$ on arbitrary networks and weakly acyclicity is violated.
	\label{thm:11JSG_IRC}
\end{theorem}

\section{Computational Hardness of Finding Optimal Placements \label{sec:hardness}}

Here, we investigate the computational hardness of computing an optimal placement, i.e., a placement where as many agents as possible are content.

\subsection{Hardness Properties for Two Types}\label{subsec:two_hardness}

We start with two types of agents and show that finding an optimal placement for the SSG in an arbitrary network $G$ is NP-hard by giving a reduction from the \textsc{Balanced Satisfactory Problem (BSP)}, which was introduced in~\cite{Gerber1998,Gerber2000} and proven to be NP-hard in~\cite{BAZGAN20061236}. This result directly implies that finding an optimal placement for the JSG with no empty nodes is NP-hard as well. 

\begin{theorem}
	Finding an optimal placement of agents for the two types SSG in a network $G$ is NP-hard for $\tau = \frac{1}{2}$. \label{thm:NP_12}
\end{theorem}

\begin{proof}
We prove the statement by giving a reduction from the BSP. Given a network $G=(V,E)$ with an even number of nodes. Let $v \in V$ and $V' \subseteq V$. We denote by $\textit{deg}_{V'}(v)$ the number of nodes in $V'$ which are adjacent to $v$. A balanced satisfactory partition exists if there is a non-trivial partition $V_1, V_2$ of the nodes $V$ with $V_1 \cup V_2 = V$, $V_1 \cap V_2 = \emptyset$ and $|V_1|=|V_2|$ such that each node~$v \in V_i$ with $i \in \{1,2\}$ has at least $\textit{deg}_{V_i}(v) \geq \frac{\textit{deg}_G(v)}{2}$, i.e., each node has at least as many neighbors in its own part as in the other.
If such a partition exists, we can find it by computing an optimal placement~$\pg^*$ in the network $G$ for two different types of agents of size $\frac{|V|}{2}$ and $\tau = \frac{1}{2}$.

The cost of a placement $\pg$ is the number of discontent agents. Obviously, a placement $\pg$ without discontent agents and thus the placement cost $\text{cost}_{\pg} (A) = 0$ is optimal. For a content agent $a \in A$ we have $\text{pnr}_{\pg}(a) \geq \frac{1}{2} = \tau$ and thus, if there are no empty nodes we know $\npg^+(a) \geq \frac{\textit{deg}_G(\pg(a))}{2}$. If we have a placement where all agents are content we can gather all nodes which are occupied by agents of type $T_1$ to the subset $V_1$ and all agents which are occupied by agents of type $T_2$ to the subset $V_2$.
It holds for every $a \in A$ that $\textit{deg}_{V_i}(\pg^*(a)) = \npg^+(a) \geq \frac{\textit{deg}_G(\pg(a))}{2} $. Hence, calculating an optimal placement must be NP-hard.
\end{proof}

\noindent The above proof relies on the fact that there are no empty nodes. The computational hardness of the JSG changes if many empty nodes exist. Obviously, it is easy to find an optimal placement if there are enough empty nodes to separate both types of agents completely and a suitable separator is known. Mapping the boundary for the transition from NP-hardness to efficient computation is a challenging question for future work. 

Next we show that finding an optimal placement is hard for high $\tau$ via a reduction from \textsc{Minimum Cut Into Equal Size (MCIES)} which was proven to be NP-hard in~\cite{GAREY1976237}.

\begin{theorem}
	Finding an optimal placement in the SSG on an arbitrary network $G = (V, E)$ with maximum node degree $\Delta_G = \max \{ \textit{deg}_G(v) \mid v \in V\}$ is NP-hard for $\tau > \frac{3\Delta_G}{3\Delta_G + 1}$.\label{thm:NP_hoch}
\end{theorem}

\begin{proof}		
We prove the statement by giving a reduction from MCIES.
	Given a network $G = (V, E)$ and an integer $W \in \mathbb{N}$.
	MCIES is the decision whether there is a non-trivial partition $V_1, V_2$ with $V_1 \cup V_2 = V$, $V_1 \cap V_2 = \emptyset$ and $|V_1|=|V_2|$ such that $|\{\{v_1,v_2\} \in V \mid v_1 \in V_1, v_2 \in V_2\}| \leq W$, i.e., there are at most $W$ edges between the two parts.
	
	Let $\Delta_G = \max \{\textit{deg}_G(v) \mid v \in V \}$ be the maximum node degree in $G$.
	We create a network $G' = (V', E')$ in which every node $v \in V$ is replaced by a clique $C_v$ in $G'$ of size $3\Delta_G + 1$.
	Each edge $\{u, v\} \in E$ will be replaced by an edge $\{u', v'\}$ between two nodes $u' \in C_u$ and $v' \in C_v$ such that each node in $G'$ has at most one neighbor outside its clique.
	Therefore, the degree of nodes in $G'$ is either $3\Delta_G$ or $3\Delta_G + 1$, and so the maximum node degree $\Delta_{G'}$ in $G'$ is $3\Delta_G + 1$.
	We have two different agent types, each consisting of $\frac{|V'|}{2}$ agents.
	Let $\tau > \frac{\Delta_{G'} - 1}{\Delta_{G'}} = \frac{3\Delta_G}{3\Delta_G + 1}$.
	Because of this, an agent is content in $G'$ if she has no neighbors of a different type.
	For a placement $p_{G'}$ to be optimal, all cliques $C$ have to be uniform, i.e.\ assign agents of the same type to each node in $C$.
	Otherwise another non-uniform clique $C'$ has to exist and we can re-assign the agents in both cliques in a placement $p_{G'}'$ to make $C$ uniform.
	In $p_{G'}$ all agent of both cliques are discontent, while in $p_{G'}'$ at least $2\Delta_G + 1$ agents in $C$ that have no neighbors outside $C$ are content.
	Since each clique is only connected to at most $\Delta_G$ other nodes, at most $2\Delta_G$ agents are discontent in $p_{G'}'$ that were content in $p_{G'}$.
	Therefore, $p_{G'}$ would not be optimal.
	
	If we have an optimal placement with $2W'$ discontent agents, we can gather all $v \in V$ where $C_v$ is occupied by agents of type $T_1$ into $V_1$, and similarly into $V_2$ for $T_2$.
	We then have $W'$ edges between the two sets $V_1$ and $V_2$.
	Hence, a placement with $2W'$ discontent agents correspond to an MCIES with $W = W'$ edges between the partitions and vice versa.
\end{proof}

\noindent For the above theorems, we used a placement cost function which counts the number of discontent agents. However, we remark that even if we change this definition into summing up the cost of all agents, i.e., $\text{cost}'_{p_G}(A) = \sum_{a \in A} \text{cost}_{p_G}(a)$, like \textit{social cost}, the above hardness results still hold. This relates to the hardness results from Elkind et al.~\cite{elkind19} which hold for the JSG with $\tau=1$ in the presence of stubborn agents which are unwilling to move.

We contrast the above results by providing an efficient algorithm for computing an optimal placement for the SSG and the JSG on a $2$-regular network with two different agent types by employing a well-known dynamic programming algorithm for \textsc{Subset Sum}~\cite{cormen,garey1979computers}. 

\begin{theorem}
	Finding an optimal placement of agents of two types in the SSG on a 2-regular network with $n$ nodes can be done in $\mathcal{O}(n^2)$ for $\tau > \frac12$. \label{thm:polytime}
\end{theorem}

\begin{proof}
Let $G=(V, E)$ be a 2-regular network, consisting of $m$ rings. Ring $i$ has $r_i$ nodes. Given a partition of the agents $P(A) = \{T_1, T_2\}$ with $|T_1| = n_1$ and $|T_2| = n_2$.

For finding a placement that minimizes $\text{cost}_{\pg}(A)$, we take the multiset $r_1, \dots, r_m$ as elements and $n_1$ as target sum as an instance of \textsc{Subset Sum}. Which we can solve in $\mathcal{O}(n^2)$ since $n_1 \leq n$. In case of a Yes-instance, we can place the agents of type $T_1$ on the rings indicated by the selected elements. Thus no agents of different types are on the same ring.  

If the instance is a No-instance, then in the optimal placement there is exactly one ring with agents of different type. This implies that at least 3 and at most 4 agents are discontent. To check if an optimal placement with 3 discontent agent is possible, we solve the \textsc{Subset Sum} instance with target sum $n_1+1$. If this is possible, then we place the $n_1$ agents on the respective rings such that exactly one node is empty. Then all empty nodes are filled with type $T_2$ agents. If the instance with target sum $n_1+1$ is a No-instance, we greedily fill the rings with consecutive type $T_1$ agents such that we get one ring with empty spots. Then we fill all the empty spots with type $T_2$ agents to obtain exactly 4 discontent agents. 

\end{proof}

\noindent Optimal placements for the JSG can be found with an analogous algorithm.

\subsection{Hardness Properties for More Types}\label{subsec:more_hardness}
Compared to the previous subsection we now show that also the number of different agent types has an influence on the computational hardness of finding an optimal placement. We establish NP-hardness even on 2-regular networks if there are sufficiently many agent types by giving a reduction from \textsc{3-Partition} which was proven to be NP-hard in~\cite{garey1979computers}.
\begin{theorem}
	Finding an optimal placement of agents of an arbitrary number of types in the $1$-$1$-SSG and $1$-$k$-SSG on a 2-regular network with $\tau > \frac12$ is NP-hard. \label{thm:2reg_hardness}
\end{theorem}

\begin{proof}

We prove the statement by giving a polynomial time reduction from \textsc{3-Partition}. Given a multiset $S$ of $3k$ positive integers.
\textsc{3-Partition} concerns whether $S$ can be partitioned into $k$ disjoint sets $S_i$ with $i \in \{1, \ldots, k\}$ of size three, such that the sum of the numbers in each subset is equal, i.e.,  $\sum_{s_i \in S_1} s_i = \sum_{s_i \in S_2} s_i = \dots = \sum_{s_i \in S_{k}} s_i$.
As these sets are disjoint, we already know that each of them sums up to $\frac{\sum_{s_i \in S} s_i}{k}$.
\textsc{3-Partition} keeps its NP-hardness if the integers in~$S$ are encoded unary.
Moreover, it remains NP-hard if we assume $\frac{\sum_{s_i \in S} s_i}{4k} < s_i < \frac{\sum_{s_i \in S} s_i}{2k}$ for all $s_i \in S$.

Based on a \textsc{3-Partition} instance, we generate a 2-regular graph, containing a ring for each $s_i \in S$ with $s_i$ nodes.
Thus our graph has $n = \sum_{s_i \in S} s_i$ nodes in total.
We can assume $s_i \ge 3$ for all $s_i \in S$, since adding a constant to all elements does not change the existence of a solution.
We now take a set of $n$ agents $A$ partitioned into types $P(A) = \{T_1, \dots, T_{k}\}$. 
Each type consists of $\frac{n}{k}$ agents. 
Assume we find an optimal placement with $\textit{cost}_{\pg}(A) = 0$ for $\tau > \frac12$. 
This means, that there is no ring that contains agents of different types, since an agent is discontent if she has a neighboring agent of different type.
Thus, we have a disjoint partitioning of the rings, such that the number of nodes in each partition adds up to $\frac{n}{k} = \frac{\sum_{s_i \in S} s_i}{k}$. 
We also assumed that $\frac{n}{4k} < s_i < \frac{n}{2k}$, thus all agents of a type $T_i$ have to be placed on exactly three rings. 
This directly implies a solution for the \textsc{3-Partition} instance.
If the corresponding \textsc{3-Partition} instance has a solution $S_1, \dots, S_k$, this produces a partitioning of the rings, such that each partition contains $\frac{\sum_{s_i \in S} s_i}{k} = \frac{n}{k}$ nodes. 
Placing the agent types according to this partitioning won't produce any ring with agents of different types on it.
Such a placement has $\textit{cost}_{\pg}(A) = 0$, which has to be optimal.

Since our reduction can be done in polynomial time for unary encoded instances of \textsc{3-partition}, this proofs the NP-hardness of finding an optimal placement.

\end{proof}

\noindent To conclude the section on the computational hardness, we want to emphasize that solving the question whether finding an optimal placement is easy or hard does not allow us to make equivalent statements for computing stable placements. The following example illustrates the rather counter-intuitive fact that an optimal placement is not necessarily stable.
\begin{theorem}
	For the SSG with two different types of agents there is a network $G$ where no optimal placement is stable. \label{thm:optnotstable} 
\end{theorem}

\begin{proof}			
	
	We prove the statement by giving an example. Consider Fig.~\ref{fig:optplacement}.
	The pictured network has two cliques $u_i$ and $v_i$  with $1\leq i \leq 3$ of size ten.
	Let $\tau > 0.9$.
	The placement $\pg^*$ depicted in Fig.~\ref{optplacement:1} has $\text{cost}_{\pg^*}(A) = 7$, and the placement $\pg$ in Fig.~\ref{optplacement:2} has $\text{cost}_{\pg}(A) = 8$.
	The former is optimal since every placement $\pg'$ other than the given two has to place agents of different types in at least one of the cliques.
	This would cause all agents in the clique to become discontent and thus yield $\text{cost}_{\pg'}(A) \ge 10$.
	However, the agents $a$ and $b$ want to swap in placement $\pg^*$.
	Hence, the unique optimal placement $\pg^*$ is not stable.
			\begin{figure}[t!]
	\centering
	\begin{subfigure}{.30\textwidth}
		\centering
		\hspace*{-0.5cm}
		\includegraphics[width=4.5cm]{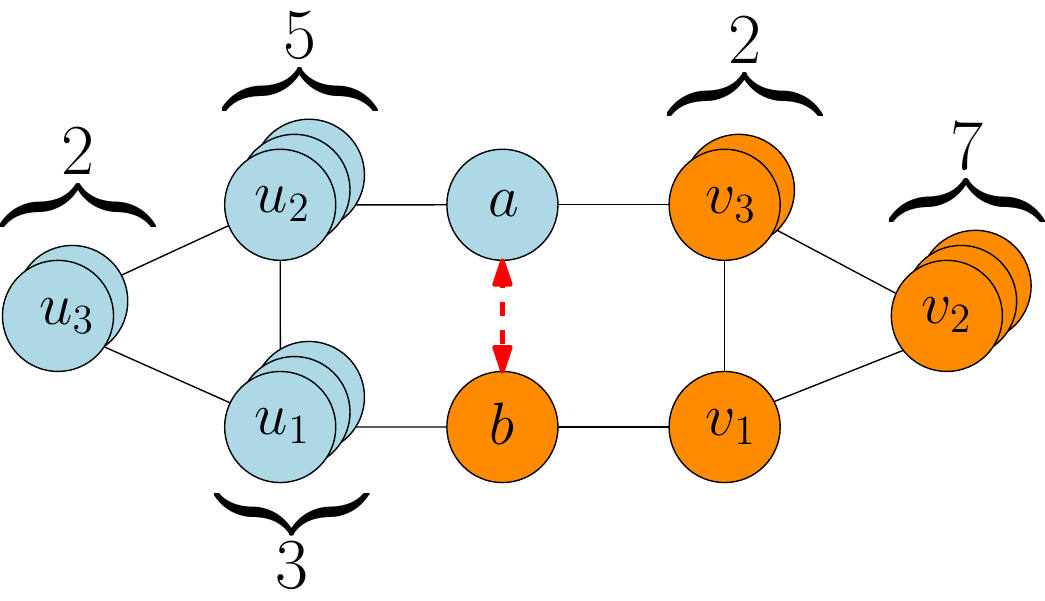}
		\caption{optimal placement $\pg^*$}
		\label{optplacement:1}
	\end{subfigure}%
	~~~~~
	\begin{subfigure}{.40\textwidth}
		\centering
		\hspace*{-1cm}
		\includegraphics[width=4.5cm]{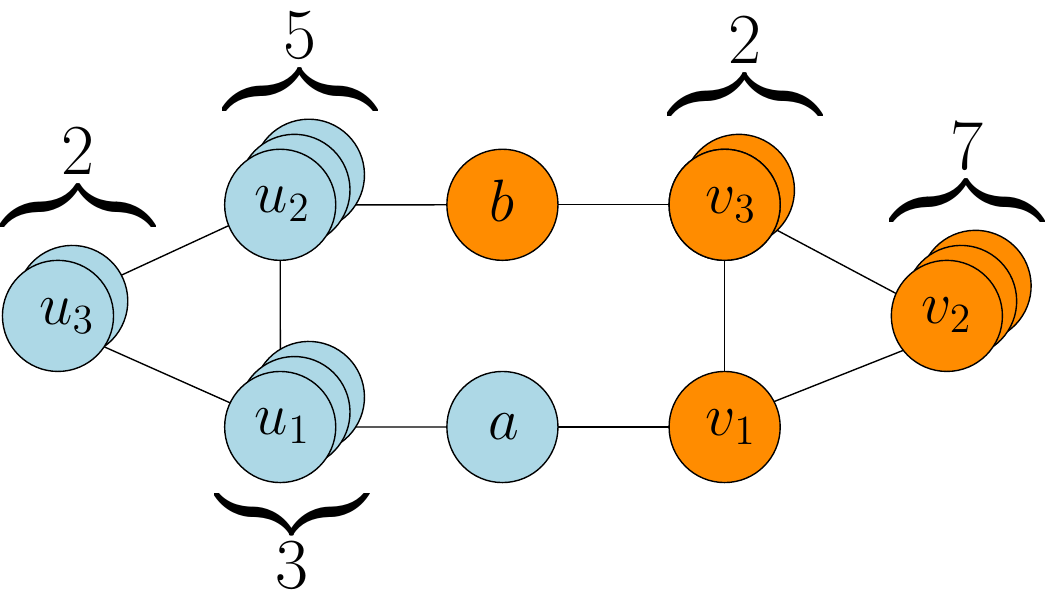}
		\caption{Placement $\pg$ after the swap}
		\label{optplacement:2}
	\end{subfigure}
	\caption{A network where the optimal placement $\pg^*$ is not in equilibrium for $\tau > 0.9$. Multiple nodes in series represent a clique of nodes of the stated size. Edges between cliques or between a clique and single nodes represent that all involved nodes are completely interconnected.}
	\label{fig:optplacement}
\end{figure}
\end{proof}

\section{Simulation}\label{subsec:Speed}
As a final aspect, we enrich our theoretical results with empirical results for the versions where IRD convergence is guaranteed. We find that for the versions with two agent types the IRD starting from uniformly random placements produce an equilibrium in $c\cdot m$ steps, where $c$ is a positive constant and $m$ is the number of edges in the underlying network. See Fig.~\ref{fig:convergence}. This meets our upper bound of $\mathcal{O}(m)$. Interestingly, IRD convergence is faster on random $8$-regular graphs than on 8-regular toroidal grids. This hints that geometry may influence the convergence speed. The details of the simulation can be found in the appendix.

	\subsection{Simulation Set-up.}\label{subsec:simsetup}
For our simulations we considered two different network topologies: toroidal grids with the Moore neighborhood, i.e., the nodes have diagonal edges and all inner nodes have degree $8$ and random $8$-regular networks.

We generated grids with $100\times 100$ up to $300\times 300$ nodes where the grid sides increased in steps of 20. To have comparable random $8$-regular graphs we generated them with the same number of nodes. For each configuration we ran the IRD starting from 100 random initial placements do derive the results depicted in Fig.~\ref{fig:convergence}.

To get the initial placements, the agents were placed uniformly at random on the nodes of the network and we assumed equal proportions of each agent type. For the jump game we used $6\%$ empty nodes. In each round the discontent agents are activated in a random order and each activated agent iterates randomly over all possible locations for a swap or a jump and chooses the first location which yields an improvement.

\begin{figure}[t!]
	\centering
	\begin{subfigure}{0.49\textwidth}
		\includegraphics[width=7.5cm]{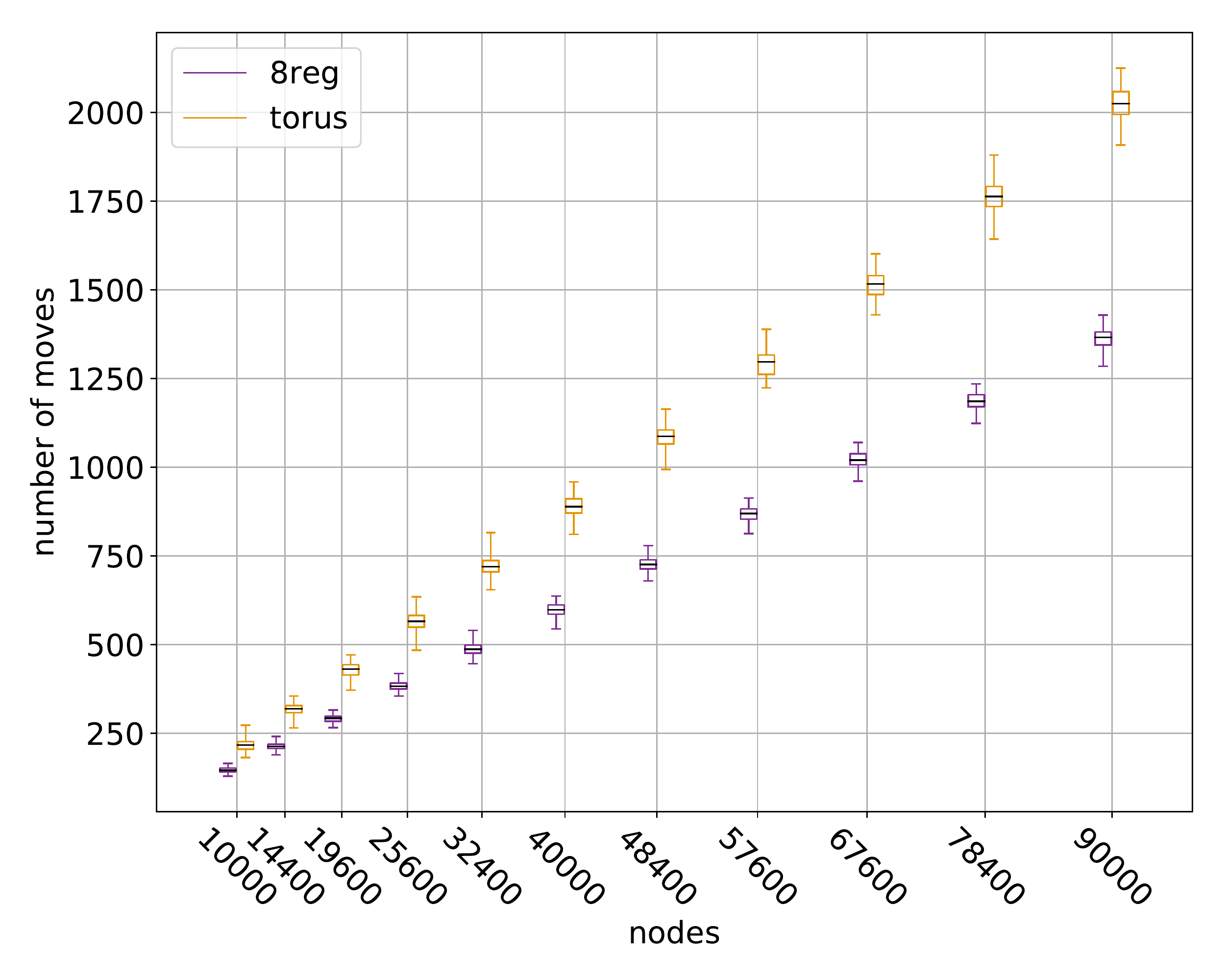}
		\caption{SSG with $k=2$ and $\tau = \frac14$}
	\end{subfigure}
	\begin{subfigure}{0.49\textwidth}
		\includegraphics[width=7.5cm]{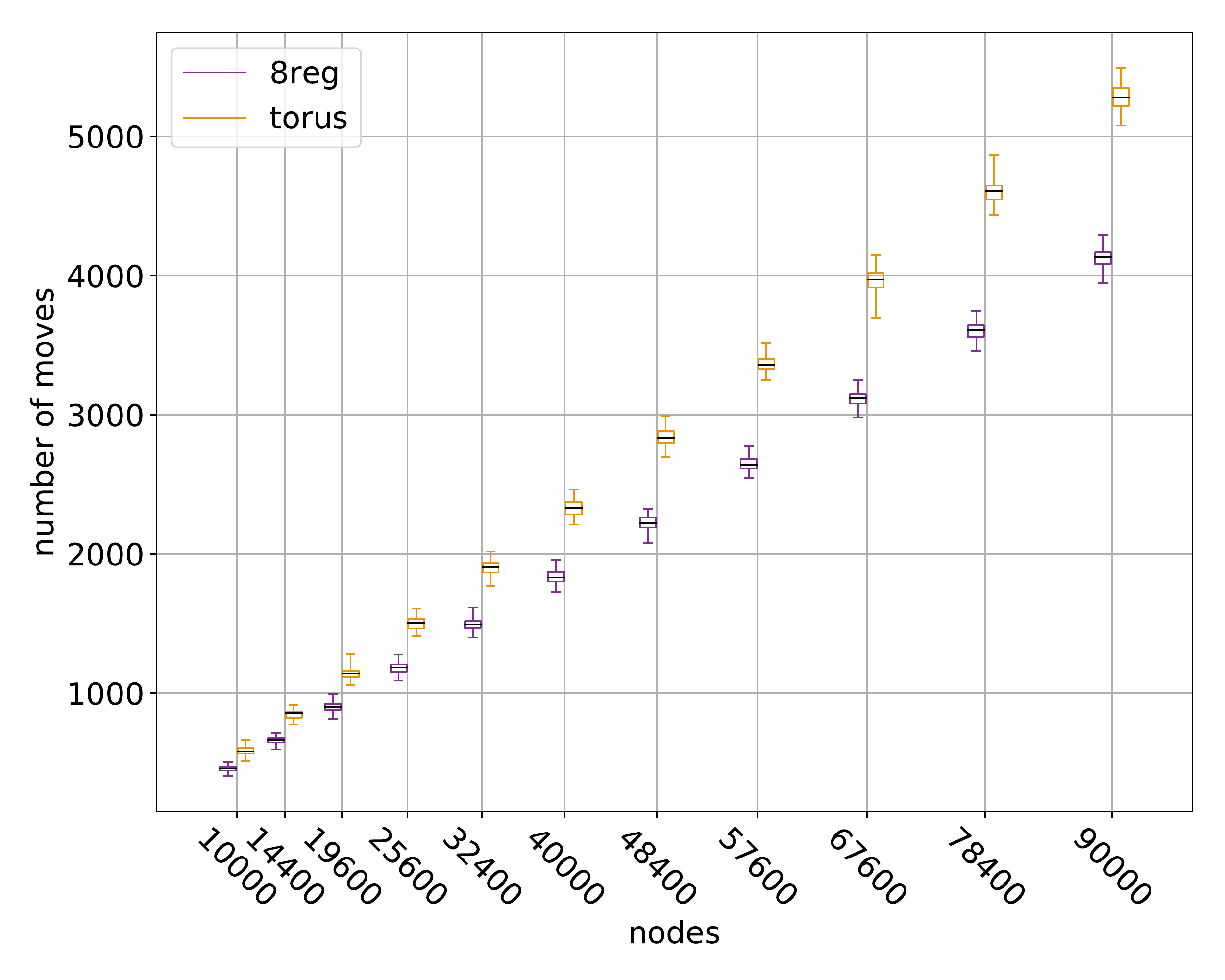}
		\caption{JSG with $k=2$ and $\tau = \frac14$}
	\end{subfigure}
	\caption{Number of moves until convergence on $8$-regular toroidal grids and $8-$regular random graphs with $10\times 10$ up to $300\times 300$ nodes over 100 trials.}
	\label{fig:convergence}
\end{figure}

\section{Conclusion and Open Questions}

We conducted a thorough analysis of the dynamic properties of the game-theoretic version of Schelling's segregation model and provided tight threshold results for the IRD convergence for several versions of the game.
Furthermore, we found that the number of agent types and the underlying graph has severe impact on the computational hardness of computing optimal placements. 

It remains open whether IRD always converge for the $1$-$1$-SSG with $\tau \in \left(\frac{1}{\Delta}, \frac{6}{\Delta}\right)$, and for the $1$-$1$-JSG with $\tau \in \left(\frac{1}{\Delta}, \frac{2}{\Delta}\right)$.
Since most versions are not guaranteed to converge via IRD, the existence of stable placements for all graph types is not given. Elkind et al.~\cite{elkind19} showed that for the $1$-$k$-JSG that stable placements exist if the underlying network is a star or a graph with maximum degree $2$ and $\tau = 1$. Furthermore they proved that if the underlying network is a tree the existence of a stable placement may fail to exist for $\tau = 1$ in the $1$-$k$-JSG. However, in general, it remains an open question in terms of different values of $\tau$ and for different underlying networks whether stable placements exist and whether they can be computed efficiently. We conjecture the following:
\begin{conjecture}
 Equilibria are not guaranteed to exist in all cases for which we constructed IRCs. 
\end{conjecture}
Also the computational hardness of finding optimal placements for some variants deserves further study and this could be extended to study the existence of other interesting states, e.g., stable states with low segregation.

Our IRD convergence results can be straightforwardly adapted to hold for the extended model by Chauhan et al.~\cite{CLM18}, where agents also have single-peaked preferences over the locations. Moreover, we are positive that also our computational hardness results can be carried over.  

Last but no least, we emphasize that there are many possible ways to model Schelling segregation with at least three agent types. For example, types could have preferences over other types which then yields a rich unexplored setting.

\bibliographystyle{abbrv}
\bibliography{schellingdynamics}

\end{document}